\documentclass[11pt]{article}
\usepackage[utf8]{inputenc}
\usepackage[T1]{fontenc}
\usepackage[english]{babel}
\usepackage[margin=2.8cm]{geometry}
\usepackage{mathpazo}
\usepackage[scaled=.95]{helvet}
\usepackage{courier}
\usepackage{amsmath,amsfonts,amssymb,amsthm}
\usepackage{graphicx}
\graphicspath{{./img/}}
\usepackage{float}
\usepackage[numbers]{natbib}
\usepackage{color}
\usepackage{hyperref}
\hypersetup{
pdfpagemode=UseOutlines,      % UseOutlines, UseThumbs, None, FullScreen : agencement au démarrage
pdfstartview=Fit,             % Fit, FitH, FitB, FitBH : vue de la page au départ (pleine largeur...)
pdffitwindow=true,            % bool: Maximiser
pdfpagelayout=TwoColumnsRight,% SinglePage, TwoColumnsRight/Left, OneColumn : affichage des pages
pdftoolbar=true,              % bool: Affichage de la barre d'outils
pdfmenubar=true,              % bool: Affichage de la barre de menus
bookmarksopen=false,          % bool: Dépliement des signets
bookmarksnumbered=true,       % bool: Numérotation des signets
colorlinks=true,              % bool: Liens colorés
pdfauthor={Ton nom},     % Auteur
pdftitle={Titre PDF},    % Titre
pdfcreator=PDFLaTeX,          % 
pdfproducer=PDFLaTeX,         %
linkcolor=blue,               % Couleur des liens
urlcolor=blue,                %             url
anchorcolor=blue,            %         du texte
citecolor=blue,              % Couleur de citation 
frenchlinks=true,             % bool: Utiliser des petites majuscules pour les liens, plutôt que de la couleur
pdfborder={0 0 0}             % Ne pas encadrer les liens
}

\usepackage[usenames,dvipsnames]{xcolor}
\usepackage[colorinlistoftodos,textwidth=2.3cm]{todonotes}
\newcommand{\cb}[1]{}

\usepackage[normalem]{ulem}
\usepackage{enumitem}

\newcommand{\dN}{\mathbb{N}}

\newcommand{\dE}{\mathbb{E}}

\newcommand{\wh}{\widehat}

\newcommand{\ind}{\mbox{1}\kern-.25em \mbox{I}}

\def\build#1_#2^#3{\mathrel{\mathop{\kern 0pt#1}\limits_{#2}^{#3}}}

\def\videbox{\mathbin{\vbox{\hrule\hbox{\vrule height1ex \kern.5em
\vrule height1ex}\hrule}}}

\numberwithin{equation}{section}
\newtheorem{thm}{Theorem}[section]
\newtheorem{rem}{Remark}[section]

\newtheorem{lem}{Lemma}[section]

\def\Var{\mathbb{V}}

\def\norm#1{\left\| #1 \right\|}
\def\abs#1{\left| #1 \right|}
\def\acc#1{\left\{ #1 \right\}}
\def\pa#1{\left( #1 \right)}
\def\pab#1{\bigl( #1 \bigr)}

\def\cro#1{\left[ #1 \right]}

\def\dt{\mbox{\rm d}t}

\def\ind{\mathbb{1}}

\def\norm#1{\left\| #1 \right\|}
\def\abs#1{\left| #1 \right|}
\def\acc#1{\left\{ #1 \right\}}
\def\pa#1{\left( #1 \right)}
\def\pab#1{\bigl( #1 \bigr)}

\def\cro#1{\left[ #1 \right]}

\def\dt{\mbox{\rm d}t}

\def\build#1_#2^#3{\mathrel{
\mathop{\kern 0pt#1}\limits_{#2}^{#3}}}
\def\tend_#1^#2{\mathrel{
\mathop{\kern 0pt\longrightarrow}\limits_{#1}^{#2}}}

\def\bkE{\mathbb{E}}

\def\det#1{\mid \! #1 \! \mid}
\newcommand{\dMah}[3]{\sqrt{\pa{#1 - #2}^T #3^{-1} \pa{#1 - #2}}}
\def\SigmaEt{\Sigma_\ast}
\def\bkE{\mathbb{E}}
\def\Xbarn{\overline{X}_n}
\def\Sigmachap{\widehat\Sigma_n}
\def\SigmaS{\Sigma_\ast}
\def\tendps{\build{\longrightarrow}_{n \rightarrow \infty}^{a.s.}}

\ifx\undefined\ly 
\gdef\beginguillemets{\leavevmode\raise0.3ex% 
         \hbox{{$\scriptscriptstyle\langle\!\langle\,$}\nobreak\ignorespaces}} 
\else 
\gdef\beginguillemets{\leavevmode\hbox{\ly(\kern-0.20em(\kern+0.20em}\nobreak} 
\fi 
\ifx\undefined\ly 
\gdef\endguillemets{\ifdim\lastskip>\z@\unskip\penalty\@M\fi 
         \leavevmode\raise0.3ex% 
         \hbox{{$\scriptscriptstyle\,\rangle\!\rangle$}}} 
\else 
\gdef\endguillemets{\nobreak\leavevmode\hbox{\kern+0.20em\ly)\kern-0.20em)}} 
\fi

\title{A mixture of ellipsoidal densities for 3D data modelling}
\author{Denis Brazey\footnote{Société Prynel, RD974, 21190 Corpeau, France},  Antoine Godichon-Baggioni\footnote{Laboratoire de Probabilités, Statistique et Modélisation, Sorbonne Université, 75005 Paris, France} and Bruno Portier \footnote{Laboratoire de Mathématiques de l'INSA, INSA ROuen-Normandie, 76800 Saint Etienne du Rouvray}
}

\date{}

\begin{document}

\maketitle

\begin{abstract}
In this paper, we propose a new ellipsoidal mixture model. This model is based a new probability density function belonging to the family of elliptical distributions and designed to model points spread around an ellipsoidal surface. 
Then, we consider a mixture model based on this density, whose parameters are estimated with the help of an EM algorithm. 
The properties of the estimates are studied theoretically and empirically. 
The algorithm is compared to a state of the art ellipse fitting method and experimented on 3D data. 
\end{abstract}

\noindent\textbf{Keywords: }
Ellipsoidal mixture model; Elliptical distribution; EM algorithm.

%%%%%%%%%%%%%%%%%%%%%%%%%%%%%%%%%%%%%%%%%%%%%%%%%%
\section{Introduction}\label{section_intro}
%%%%%%%%%%%%%%%%%%%%%%%%%%%%%%%%%%%%%%%%%%%%%%%%%%
Fitting geometric primitives to a set of noisy points holds significant practical relevance across a multitude of scientific domains, including environmental science \cite{Burt19}, agriculture \cite{Parr22}, computer-aided design (CAD) \cite{Romanengo22}, industrial applications \cite{Nahangi19}, robotics \cite{Shi21}, autonomous vehicles \cite{Du18}, and LiDAR applications \cite{Shi21b}.
% general applications   

In particular, geometric modeling offers a powerful means of representing objects present in images captured by imaging sensors. These sensors provide 2D or 3D point sets that describe the external surfaces of objects. Describing these data using simple geometric shapes such as spheres, cylinders, planes, and ellipsoids enables us to summarize and extract meaningful information. This high-level representation is crucial for efficient analysis and understanding of image content.
% why its usefull 

In practice, two fundamental scenarios arise: shape detection and shape fitting. In classical fitting problems, it is assumed that the majority of points belong to the shape of interest. Consequently, the entire data set is modeled with a single instance of the shape, with some points considered as outliers. The accuracy of this modeling is inevitably impacted by the noise in the data points \cite{Li20}, as well as their susceptibility to outliers and high levels of missing data. Typically, these algorithms are grounded in the principles of least squares, where the minimized distance is usually algebraic or geometric \cite{Li04, Kesaniemi17, Ahn02}.

Conversely, in shape detection problems, the data may contain a high proportion of outliers that belong to other objects. The objective here is to extract one or multiple instances of a given shape without prior knowledge of which observations belong to them. Techniques like RANSAC and Hough transform are commonly employed to address this challenge. RANSAC randomly selects sets of inlier points, fits the primitive shape, and retains those closest to the data \cite{Sun2019}. The Hough transform, on the other hand, employs a voting scheme in a discretized parameter space \cite{Sommer20}. However, these methods often exhibit sensitivity to parameter tuning. An alternative approach is to first apply a segmentation algorithm, followed by a fitting algorithm on the detected clusters. Recently, deep neural networks have also been employed for primitive detection, segmenting point clouds into clusters and predicting classification and membership scores for subsequent fitting steps \cite{Qi17, Li19, Sharma20}.

% 
% detection 

% mixture models 
Mixture models \cite{McLachlan2000} are statistical tools that can efficiently classify data and estimate a probabilistic model for each component.  
They find applications in various image processing tasks, including background subtraction \cite{Zuo19}, image segmentation \cite{Yang20}, head detection \cite{Brazey14}, and ellipsoid fitting \cite{Zhao21}.

In this work, 
our focus lies in modeling a 3D point cloud using a mixture of ellipsoidal shapes. The primary challenge is selecting the component density within the mixture. We propose employing a parametric density of the form:

$$f_{\theta}(x) = C \, \exp\pa{- \dfrac{d_{\theta}^2(x)}{2 \sigma^2}}, \quad x \in \mathbb{R}^d$$  

\noindent
where $C$ represents the normalization constant, $d_{\theta}(x)$ quantifies the signed distance between point $x \in \mathbb{R}^d$ and the surface $S$ defined by the parameter vector $\theta$, and $\sigma$ controls the noise level. This distribution characterizes points distributed around the surface $S$ of the considered shape with fluctuations in the normal direction. The surface is defined by the zero level line of the distance function $d_{\theta}$, where the density $f$ reaches its maximum value.

To fit an ellipsoid with a center $\mu \in \mathbb{R}^d$ and shape matrix $\Sigma \in \mathbb{R}^{d \times d}$, we propose using the following signed distance:
\[
d_\theta ( x) = \sqrt{(x - \mu)^T \, \Sigma^{-1} \, (x - \mu)} - 1, 
\]
\noindent resulting in a new pdf and leading to a new mixture model based on this density. 
Let us note that for the choice $d_{\mu}(x) = \norm{x - \mu}$, where $\norm{\cdot}$ denotes the Euclidian norm in $\mathbb{R}^d$, then $S = \left\lbrace \mu \right\rbrace$ and the density $f$ reduces to the usual Gaussian density of parameters $\mu$ and $\sigma$. 
When considering $d_{(\mu, r)}(x) = \norm{x - \mu} - r$, then $S$ is the sphere of center $\mu \in \mathbb{R}^d$ and radius $r > 0$ and we obtain the density  introduced in \cite{Brazey14}.

The primary objectives of this study are twofold: first, to develop efficient methods for parameter estimation of this new density, providing direct estimates and their convergence rates; second, to propose iterative estimation methods with improved behavior, particularly relevant for estimating the parameters of the mixture model. Finally, we estimate the parameters of this mixture model using the EM algorithm \cite{Dempster77}, where the maximization step is obtained thanks to the aforementioned iterative estimates. Finally, this approach seems to be confirmed by  experiments on real and  simulated data.   

The structure of the paper is as follows: Section \ref{section_generalisation} introduces the ellipsoidal mixture model, while Section \ref{section_estimation} delves into the estimation of model parameters. Section \ref{section_experiments_simul} presents results obtained from simulated data, and Section \ref{section_applications} discusses experiments conducted on real data. The proofs and additional results are deferred to the Appendix.
% outline  

%%%%%%%%%%%%%%%%%%%%%%%%%%%%%%%%%%%%%%%%%%%%%%%%%%
\section{The ellipsoidal mixture model}\label{section_generalisation}
%%%%%%%%%%%%%%%%%%%%%%%%%%%%%%%%%%%%%%%%%%%%%%%%%%

In this section, we introduce the new ellipsoidal density and establish some usefull properties and results which will be used in the rest of the paper. 

%%%%%%%%%%%%%%%%%%%%%%%%%%%%%%%%%%%%%%%%%%%%%%%%%%
\subsection{The new ellipsoidal density}\label{ssection_density}
%%%%%%%%%%%%%%%%%%%%%%%%%%%%%%%%%%%%%%%%%%%%%%%%%%

In order to model points spread around an ellipsoidal surface, we propose to use a parametric density of the form 

\begin{equation}\label{eq_gen_density}
f_{\theta}(x) = C \, \exp\pa{- \dfrac{d_{\theta}^2(x)}{2 \sigma^2}}, x \in \mathbb{R}^d
\end{equation} 

\noindent 
where $C$ is the normalization constant, $d_{\theta}(x)$ quantify the signed distance between the point $x \in \mathbb{R}^d$ 
and the surface $S$ of parameter vector $\theta$ 
and $\sigma$ controls the level of noise.

% definition de l ellipsoide 
Let us discuss the choice of the distance function. 
A $d$-dimensional ellipsoid is characterized by a center $\mu \in \mathbb{R}^d$ and a shape matrix $\Sigma \in \mathbb{R}^{d \times d}$ symetric and positive definite. 
The ellipsoidal surface $S$ is defined as follows 
\begin{equation} \label{def_ellipsoid}
S = \left\lbrace x \in \mathbb{R}^d, \pa{x - \mu}^T \, \Sigma^{-1} \, \pa{x - \mu} = 1\right\rbrace. 
\end{equation}

% signification des parametres 
The center $\mu$, defining the location of the ellipsoid, is the intersection between the $d$ axes of the ellipsoid. 
The shape matrix $\Sigma$ defines the shape and the orientation of the surface. 
More precisely, we can consider the decomposition  $\Sigma = P \, D \, P^{-1}$, where $D$ is the diagonal matrix of eigenvalues $\pa{\lambda_1, \ldots, \lambda_d}$ of $\Sigma$ such that $0 < \lambda_1 = \lambda_{min}(\Sigma) \leq \ldots \leq \lambda_d = \lambda_{max}(\Sigma)$ and $P$ is the orthogonal matrix whose colums are the associated eigenvectors $\pa{v_1, \ldots, v_d}$. 
Principal axes of the ellipsoid are given by $\pa{v_i}_{i=1}^d$ and their lengths by $\pa{2 \, \sqrt{\lambda_i}}_{i=1}^d$.  Since the Euclidian distance is difficult to write analytically in the case of the ellipsoid, we decide to use the Mahalanobis distance defined as 
\begin{equation}
d_m\pa{x, \mu, \Sigma} = \sqrt{(x - \mu)^T \, \Sigma^{-1} \, (x - \mu)}, 
\end{equation}

\noindent
where $x \in \mathbb{R}^d$ is the considered point and $\theta = \pa{\mu, \Sigma}$ is the parameter of the ellipsoid. 
From \eqref{def_ellipsoid}, the ellipsoid is defined by the contour line of value one.  
%Figure \ref{contour_ellipse} shows the contour lines for a given ellipsoid. 
Note that in the particular case where $\Sigma = I_{d}$, the Mahalanobis distance coincides with the Euclidian distance.
Finally, the general density \eqref{eq_gen_density} rewrites 
\begin{equation} \label{def_densite}
f(x) = C_d \, \exp\pa{- \dfrac{1}{2 \, \sigma^2} \pa{\dMah{x}{\mu}{\Sigma} - 1}^2}, 
\end{equation}

\noindent 
where $\mu \in \mathbb{R}^d$, $\Sigma \in \mathbb{R}^{d \times d}$ is symetric and positive definite and $\sigma^2 > 0$. 
The parameter $\sigma$ controlls the dispersion of the observations around the surface. 
The normalization constant $C_d$ equals to 
\[
C_d = \dfrac{\Gamma(d / 2)}{2 \, \pi^{d/2} \det{\Sigma}^{1/2} \, J_{d-1}(\sigma)}, 
\]
where $\Gamma$ denotes the Gamma function, and for $q \in \mathbb{N}$ and $\alpha > 0$, 
\[
J_q(\alpha) = \int_{0}^{\infty} t^q \, \exp\pa{- \dfrac{(t - 1)^2}{2 \, \alpha^2}} \, dt. 
\]
% modelisation 
The detailed calculus are given in Appendix \ref{appendixA_ell}.
As this distribution is designed to model points spread around an ellipsoid, we only consider the case where $\sigma$ is small enough. 
More precisely, we suppose that the value of $\sigma$ is largely smaller than the length of the smallest half-axis of the ellipsoid, namely $\sqrt{\lambda_{min}(\Sigma)}$. 
This assumption means that fluctuations do not interact with the other side of the surface. 
On top of that, if this condition is not verified, the density will model points located near the center of the ellipsoid, which is not our objective.  
The ellipsoidal density of parameters $\mu = (0, 0)^T$, $\Sigma = \left( \begin{array}{cc} 4 & 0 \\ 0 & 1 \end{array} \right)$ and $\sigma^2 = 0.1$ is represented Figure \ref{density_2D}(a) and a sample from the distribution is given Figure \ref{density_2D}(b). 
The sample illustrates that fluctuations are gaussians oriented in the normal direction of the surface.

\begin{figure}[h!]\centering
\begin{tabular}{cc}
\includegraphics[scale=0.23]{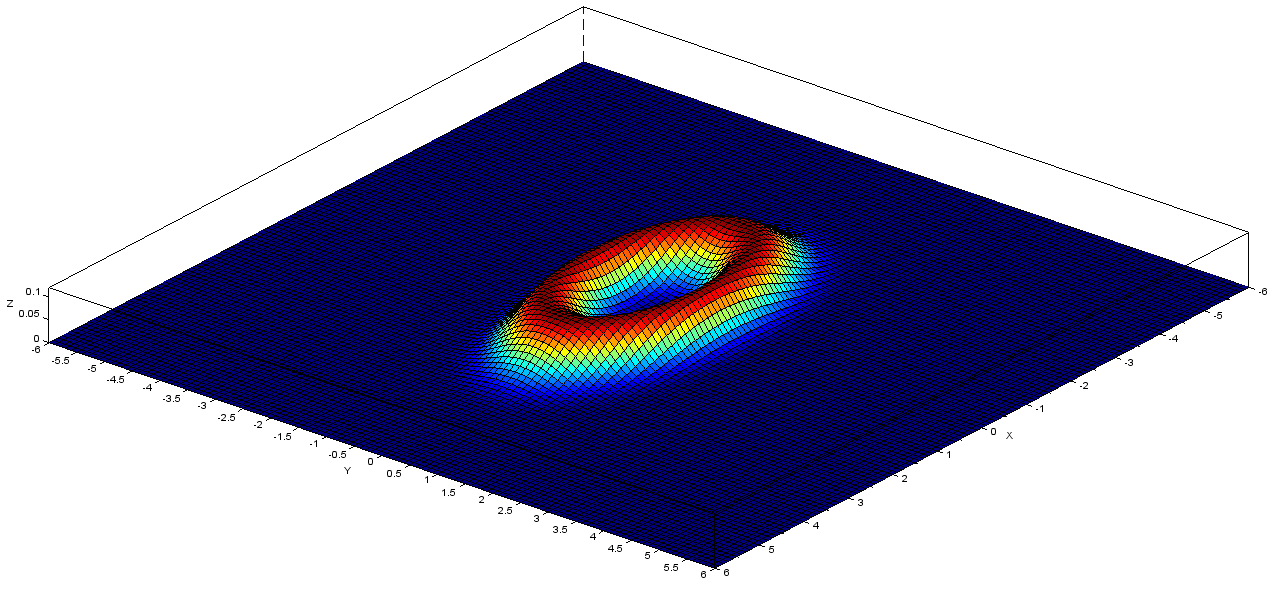} &
\includegraphics[scale=0.24]{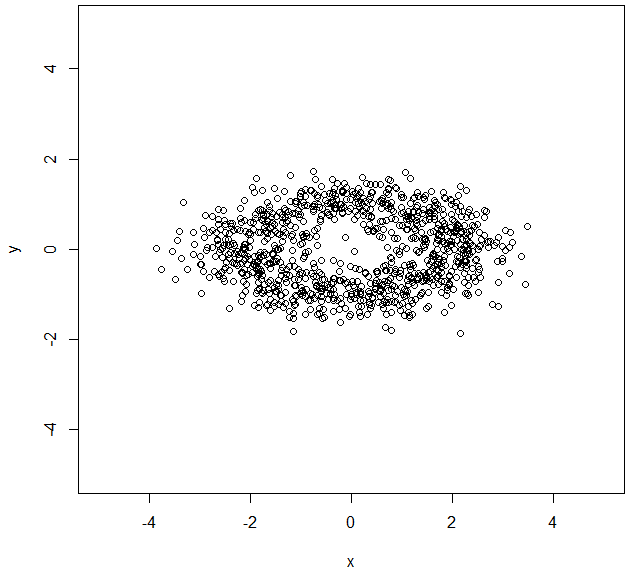} \\
(a)& (b) 
\end{tabular}
\caption{\small In (a) the ellipsoidal density in the case $d=2$ and in (b) a sample of size $n=1000$ from the distribution.}
\label{density_2D}
\end{figure}

\normalsize
% distributions elliptiques et ecriture polaire 
We now give some useful properties and results concerning the new distribution. 
Let us remark that in the case $d \geq 2$, our ellipsoidal distribution belongs to the family of elliptical distributions. 
%When $d \geq 2$, the pdf $f$ belongs to the elliptical density family. 
Indeed, if we consider the function 
\[
g_{\alpha}(t) = \dfrac{\Gamma(d/2)}{2 \pi^{d/2}J_{d-1}(\alpha)} \, \exp\pa{- \dfrac{\pa{\sqrt{t} - 1}^2}{2 \alpha^2}}, 
\]
then the density $f$ rewrites 
\[
f(x) =\ \det{\Sigma}^{-1/2} g_{\sigma}\pa{(x - \mu)^T \, \Sigma^{-1} \, (x - \mu)}, 
\]
which means from Definition 1.5.2 from \cite{Muirhead1982} that $f$ is an elliptical pdf. 
%the pdf of an elliptical distribution. 
Consequently, considering a random vector $X$ of density $f$ and $Y = \Sigma^{-1/2}\pa{X - \mu}$, then $Y$ has density $g_{\sigma}(y^T \, y)$ and can be written under the form $Y = W \, U$, where $W$ and $U$ are independent, $U$ is uniformly distributed on the d-dimensional unit sphere and $W$ has the density 
\[
\varphi(t) = \dfrac{2 \pi^{d/2}}{\Gamma(d/2)} \, t^{d-1} , \quad  \quad  g_{\sigma}(t^2) = \dfrac{t^{d-1}}{J_{d-1}(\sigma)} \, \exp\pa{- \dfrac{(t - 1)^2}{2 \sigma^2}} \, \mathbf{1}_{t \geq 0}. 
\]
The random vector $X$ can therefore be written 
\[
X = \mu + \Sigma^{1/2} \, W \, U. 
\]
% 
%Let us now give in the following theorems the first moments of $X$ and some other results usefull for justifying the properties of some estimators introduced in section \ref{section_estimation}. 
Let us now focus on the first moments of $X$ and some other results which will be usefull to prove some estimators properties in Section \ref{section_estimation}. 

\begin{lem} \label{thm_ellipsoid}
Let $X$ be a random vector of $\mathbb{R}^d$ (with $d \geq 2$) of pdf $f$ introduced in \eqref{def_densite}. 
For $q\geq 1$, we set $J_q = J_q(\sigma)$. 
Then, 
\[
\dE\cro{X} = \mu \qquad \mbox{ and } \qquad 
\Var\cro{X} = \dE\cro{\pa{X - \mu}\pa{X - \mu}^T} = 
\dfrac{J_{d+1}}{d \, J_{d-1}} \, \Sigma = d^{-1} \SigmaEt
\]
where we set $\displaystyle \SigmaEt = \dfrac{J_{d+1}}{J_{d-1}}\,\Sigma$.
In addition,
\[
\Var\cro{\dMah{X}{\mu}{\SigmaEt}}\ =\ \dfrac{J_{d+1}J_{d-1} - J_d^2}{J_{d+1}J_{d-1}}
\ = \ \widetilde\sigma^2 .
\]
\end{lem}

The proof is given in Section \ref{sec::proofs::lemma}.

\begin{rem} \label{remq_thm_ellisoid}
\rm{

In the rest of the paper, we suppose that $\sigma \leq 0.02$, which is a reasonable value for our experimental settings (cf Appendix \ref{appendixA_bis_ell}). 
Results of the lemma are formulated for any dimension $d \geq 2$ and for any value of $\sigma$. 
In the particular case where $d=3$ and $\sigma$ is small enough, we can approximate  
$J_2$ by $C \, \pa{1 + \sigma^2}$, $J_3$ by $C \, \pa{1 + 3 \sigma^2}$ and $J_4$ by $C \, \pa{1 + 6 \sigma^2 + 3 \sigma^4}$, where $C = \sigma \, \sqrt{2 \pi}$.

Consequently, one can approximate $\SigmaEt = (1 + 6\sigma^2 + 3\sigma^4) \, (1 + \sigma^2)^{-1} \, \Sigma$ and $\tilde{\sigma}^2 = (1 + 3 \sigma^4) \, (1 + 7 \sigma^2 + 9 \sigma^4 + 3 \sigma^6)^{-1} \, \sigma^2$. 
Finally, as $\sigma$ is assumed sufficently small, we deduce that $\SigmaEt \simeq \Sigma$ and $\tilde{\sigma}^2 \simeq \sigma^2$ (cf. Appendix \ref{appendixA_ell} and Appendix \ref{appendixA_bis_ell}) . 
}
\end{rem}

The following lemma gives results that will improve the estimation of parameter $\mu$ via a Back-Fitting (BF for short) type algorithm. 

\begin{lem} \label{tthm_Xetoile}
Let $X$ be a random vector of $\mathbb{R}^d$ (with $d \geq 2$) of pdf $f$ introduced in \eqref{def_densite}. 
Denoting
\[
X^\ast :=  X - \dfrac{(X - \mu)}{\dMah{X}{\mu}{\SigmaEt}},
\]
then  
\[
 \dE\cro{X^\ast} = \mu \qquad \qquad \mbox{ and } \qquad \qquad \Var\pa{X^\ast} = \dE\cro{\pa{W - \sqrt{\dE\cro{W^2}}}^2} \dfrac{\Sigma}{d},
\]
and $\norm{\Var{\cro{X^\ast}}} < \norm{\Var\cro{X}}$. 
% puisque $J_{d+1} - 2J_> 0$. 
\end{lem}
The proof is given in Section \ref{sec::proofs::lemma}.
Lemma \ref{thm_ellipsoid} and \ref{tthm_Xetoile} combined with Remark \ref{remq_thm_ellisoid} will be particularly usefull for the estimation of the parameters of the distribution from a sample (see Section \ref{section_estimation}). 

%%%%%%%%%%%%%%%%%%%%%%%%%%%%%%%%%%%%%%%%%%%%%%%%%%
\subsection{The finite mixture model}\label{ssection_mm}
%%%%%%%%%%%%%%%%%%%%%%%%%%%%%%%%%%%%%%%%%%%%%%%%%%

A finite mixture model is a probabilistic model characterised by a density function $h$ of the form 
\[
h(x | \Theta) = \sum_{k=1}^K \pi_k \, f(x | \theta_k), 
\]
where $f$ belongs to a parametric density family, $K$ is the number of components of the mixture, $(\pi_k)$ are the mixing weights satisfying $\pi_k \geq 0$ and $\sum_{k=1}^K \pi_k = 1$ and $\Theta = \pa{\pi_1, \ldots , \pi_K, \theta_1, \ldots  , \theta_K}$ is the vector containing all the parameters of the mixture. 
%The density $h$ is a convex combination of $K$ densities, each specified through a parameter vector. 
% 
%A finite mixture model is a probabilistic model whose density function $h$ is a convex combination of $K$ density functions. 
%Generally, these densities belong to a parametric density family and each component is specified through the parameter of the density.
%More precisely, if $f$ denotes the common parametric density and $\theta_k$ is the parameter associated with the component $k$, then the mixture density is of the form 
%\begin{eqnarray} \label{eq_mixturemodel}
%h(x | \Theta) = \sum_{k=1}^K \pi_k \, f(x | \theta_k)
%\end{eqnarray}

%\noindent 
%where $(\pi_k)$ are the mixing proportions satisfying $\pi_k \geq 0$ and $\sum_{k=1}^K \pi_k = 1$ and $\Theta = \pa{\pi_1, \ldots , \pi_K, \theta_1, \ldots  , \theta_K}$ is the vector containing all the parameters of the mixture.
%\vspace{1ex}

% 
The choice of the common density $f$ depends on the data set. 
In many applications, Gaussian or Poisson densities have been used. 
In our case, since our data are 3-dimensional points spread around ellipsoidal surfaces, the pdf $f$ is set to the new density introduced in \eqref{def_densite} 
and the parameter vector $\theta_k$ is then defined as $\theta_k = \pa{\mu_k, \Sigma_k, \sigma_k^2}$. 

%%%%%%%%%%%%%%%%%%%%%%%%%%%%%%%%%%%%%%%%%%%%%%%%%%
\section{Parameters estimation}\label{section_estimation}
%%%%%%%%%%%%%%%%%%%%%%%%%%%%%%%%%%%%%%%%%%%%%%%%%%

% 
This section concerns the estimation of the unknown parameters of the  model from a sample in the case $d=3$ which is the major concern of the paper. 
Let us note that the estimation of the model in the case $d=2$ is straightforward. 
For shake of clarity, sections \ref{ssection_directestim} and \ref{ssection_iterestim} are dedicated to the estimation of the parameters of the ellipsoidal density $f$ while section \ref{ssection_estimmm} deals with the estimation of the ellipsoidal mixture model of density $h$. 

%%%%%%%%%%%%%%%%%%%%%%%%%%%%%%%%%%%%%%%%%%%%%%%%%%
\subsection{Direct estimation of the parameters}\label{ssection_directestim}
%%%%%%%%%%%%%%%%%%%%%%%%%%%%%%%%%%%%%%%%%%%%%%%%%%

% 
Let $X_1, \ldots, X_n$ be a sample of independent and identically distributed  random vectors of $\mathbb{R}^3$ with density $f$ introduced in \eqref{def_densite} with $d=3$ and under the assumption that $\sigma \ll \sqrt{\lambda_{min}(\Sigma)}$. 
%The objective is to find some estimators of the parameters $\mu$, $\Sigma$ and $\sigma^2$ characterizing the distribution. 
The objective is to estimate the parameters $\mu$, $\Sigma$ and $\sigma^2$ characterizing the distribution.

Let us start with the estimation of the center $\mu$. 
To this aim, we consider the empirical mean of the sample 
\[
\overline{X}_n = \dfrac{1}{n} \sum_{i=1}^n X_i.
\]
It is clear from Theorem \ref{thm_ellipsoid} that $\Xbarn$ is an unbiased estimator of $\mu$ and its rate of convergence derived from the law of the iterated logarithm, is given in  Theorem \ref{thm_conv_direct}. 
We now consider the estimation of the shape matrix $\Sigma$.
We recall for this purpose the usual covariance matrix estimator 
\[
S_n = \dfrac{1}{n} \sum_{i=1}^n \pa{X_i - \mu}\pa{X_i - \mu}^T .
\]
Thanks to the results of Theorem \ref{thm_ellipsoid} and the strong law of large numbers, 
we can easily show that $S_n$ is an unbiased and convergent estimator of 
$J_4(\sigma) \, \Sigma / (3\,J_2(\sigma))$. 
When $\sigma$ is sufficiently small, the ratio $J_4 \, J_2^{-1}$ 
is close to $1$ (see Remark \ref{remq_thm_ellisoid}) and in that case a reasonable estimator of $\Sigma$ is $\Sigma_n = 3 \, S_n$. 
However, the parameter $\mu$ being unknown, it is estimated by $\Xbarn$
and we finaly propose to estimate $\Sigma$ by
\[
\widehat{\Sigma}_n\ =\ 
\dfrac{3}{n} \sum_{i=1}^n \pa{X_i - \Xbarn}\pa{X_i - \Xbarn}^T .
\]
By using the decomposition 
\[
\widehat{\Sigma}_n = \Sigma_n  - 3 \, \pa{\Xbarn - \mu}\pa{\Xbarn - \mu}^T , 
\]
we can show that $\widehat{\Sigma}_n$ converges almost surely towards  $\Sigma^* = J_4(\sigma) \, J_2^{-1}(\sigma) \, \Sigma$ with rate (see Theorem \ref{thm_conv_direct}) and consequently when $\sigma$ is small, 
$\widehat{\Sigma}_n$
is a reasonable estimator of $\Sigma$.

We now focus on the estimation of parameter $\sigma^2$. 
In order to estimate this parameter, we introduce the estimator  
\[
\sigma_n^2 = \dfrac{1}{n} \sum_{j=1}^n \pa{\xi_j -  \bkE\cro{\xi_j}}^2
\]
where for $j=1, \ldots, n$, 
$\xi_j = \sqrt{\pa{X_j - \mu}^T (\Sigma^*)^{-1} \pa{X_j - \mu}}$.
From the strong law of large numbers and from results of Theorem \ref{thm_ellipsoid}, 
we can deduce that $\sigma_n^2$ converges almost surely towards   
$\widetilde\sigma^2 \simeq \sigma^2$ when $\sigma$ is sufficiently small (Remark \ref{remq_thm_ellisoid}). 
However, this estimator depends on the unknown parameters $\mu$ and $\Sigma^*$.
We therefore propose to estimate $\sigma^2$ by  
\[
\widehat\sigma_n^2 = \dfrac{1}{n} \sum_{j=1}^n \widehat\xi_j^2 - 
\pa{\dfrac{1}{n} \sum_{j=1}^n \widehat\xi_j}^2
\]
where for $j=1, \ldots, n$, 
$\widehat\xi_j = \sqrt{\pa{X_j - \Xbarn}^T \widehat\Sigma_n^{-1} \pa{X_j - \Xbarn}}$.
%We can prove the almost sure convergence of this estimator towards $\widetilde\sigma^2$ with rate (see Theorem \ref{thm_conv_direct}). 
%Let us note that the estimation of $\sigma^2$ requires an estimation of $(\Sigma_*)^{-1}$. 
%In the next section, we will explain how to iteratively estimate $(\Sigma_*)^{-1}$ without inverting $\widehat\Sigma_n$.
Convergence results are given in the following theorem. 
\begin{thm} \label{thm_conv_direct}
Let $X_1, \ldots, X_n$ be a sample of independent and identically distributed random vectors of $\mathbb{R}^3$ with pdf $f$. Then, we have 
$$\parallel\!\Xbarn - \mu\!\parallel\ =\ O\pa{a_n}, \quad 
\norm{\widehat{\Sigma}_n -\Sigma_*}\ =\ 
O\pa{a_n} \, \mbox{ and} \quad
\abs{\wh{\sigma}^2_n  - \widetilde{\sigma}^2} =  O(a_n)
\ \  \mbox{a.s.}$$

where $a_n = \sqrt{(\log\!\log n) / n}$.
\end{thm}
The proof is given in Appendix \ref{appendixB_ell}.

%%%%%%%%%%%%%%%%%%%%%%%%%%%%%%%%%%%%%%%%%%%%%%%%%%
\subsection{Iterative estimation of the parameters}\label{ssection_iterestim}
%%%%%%%%%%%%%%%%%%%%%%%%%%%%%%%%%%%%%%%%%%%%%%%%%%

% 
In this section, we propose a more efficient estimation strategy, which will be of particular interest for estimating the parameters of mixture models. 
The estimate $\overline{X}_n$ is an unbiased and convergent estimator of $\mu$. 
However, this estimator has a much larger variance than another unbiased and convergent estimator of $\mu$ denoted as $\overline{X}_n^\ast$ and defined by  
\[
\overline{X}_n^\ast = \overline{X}_n - \dfrac{1}{n} \sum_{i=1}^n \dfrac{\pa{X_i - \mu}}{\dMah{X_i}{\mu}{\SigmaEt}}.
\]
Indeed, using the results of Theorem \ref{tthm_Xetoile}, we can show that $\Xbarn^\ast$ is an unbiased and convergent estimator of $\mu$ satisfying 
$\| \Var (\overline{X}_n^\ast) \| < \| \Var (\overline{X}_n) \|$.
Nevertheless, the estimator $\overline{X}_n^\ast$ depends on unknown parameters $\mu$ and 
$\SigmaEt$ which must be estimated.
We therefore propose the following strategy to improve the estimation of parameters 
$\mu, \Sigma$ and $\sigma^2$.
First, we estimate $\mu$ by 
$\widehat{\mu}_n^{(0)} = \Xbarn$, then $\Sigma_\ast$ 
by $\widehat\Sigma_n^{(0)} = \widehat\Sigma_n$.
Then, we improve the estimation of $\mu$ by considering 
\begin{equation} \label{Def_muchap1}
\widehat{\mu}_n^{(1)} = \overline{X}_n - \dfrac{1}{n} 
\sum_{i=1}^n \dfrac{\pa{X_i - \widehat{\mu}_n^{(0)}}}
{\dMah{X_i}{\widehat{\mu}_n^{(0)}}{\pa{\widehat\Sigma_n^{(0)}}}}
\end{equation}
and the estimation of $\Sigma_\ast$ by considering 
\begin{equation} \label{sigma_1}
\widehat{\Sigma}_n^{(1)} = \dfrac{3}{n}\, \sum_{i=1}^n \pa{X_i - \widehat{\mu}_n^{(1)}} 
\pa{X_i - \widehat{\mu}_n^{(1)}}^T .
\end{equation}

The rate of convergence of these new estimates can be derived from Theorem \ref{thm_conv_direct}. In addition, 
the process \eqref{Def_muchap1}-\eqref{sigma_1} can be iterated until convergence to improve the estimations of $\mu$ and $\SigmaEt$. 
At the end of each iteration, we update the estimates as follows : $\widehat{\mu}_n^{(0)} = \widehat{\mu}_n^{(1)}$ and $\widehat{\Sigma}_n^{(0)} = \widehat{\Sigma}_n^{(1)}$.   
Although we did not succeed in proving that the backfitting estimates are theoretically better, this seems to be confirmed in practice (see Section \ref{section_experiments_simul}). 
%The convergence of this iterative backfitting strategy is not theoretically established but has been verified in simulations. 
Finally, the parameter $\sigma^2$ can then be estimated by  
\[
\widehat\sigma^2_n = \dfrac{1}{n}\,\sum_{j=1}^n \widetilde\xi_j^2
- \pa{\dfrac{1}{n}\,\sum_{j=1}^n \widetilde\xi_j}^2
\]
with $\displaystyle\widetilde\xi_i  =  
\dMah{X_i}{\widehat{\mu}_n^{(1)}}{\pa{\widehat\Sigma^{(1)}_n}}$. 

%%%%%%%%%%%%%%%%%%%%%%%%%%%%%%%%%%%%%%%%%%%%%%%%%%
\subsection{Estimation of the mixture model parameters} \label{ssection_estimmm}
%%%%%%%%%%%%%%%%%%%%%%%%%%%%%%%%%%%%%%%%%%%%%%%%%%

% 
This section concerns the estimation of a mixture model with $K > 1$ components where the component $k$ is described by the density 

\begin{equation*}
f(x \mid \theta_k) = \frac{(2 \pi)^{-3/2}\det{\Sigma_k}^{-1/2}}{2 \, \sigma_k (1 + \sigma_k^2)} \, \exp\pa{- \dfrac{1}{2 \, \sigma_k^2} (\dMah{x}{\mu_k}{\Sigma_k} - 1)^2}
\end{equation*}
where $\theta_k = (\mu_k, \Sigma_k, \sigma_k)$ and $\sigma_k$ is supposed sufficiently small. 
%Let us note that we have taken the approximated value of $C_3$ instead of its exact expression. 
Let us note that the approximated version of $C_3$ is used instead of its exact expression. 
%This approximation is reasonable when $\sigma_k$ is small enough and is usefull to simplify computations. 

% 
%Let $\textbf{x} = \pa{x_1, \ldots, x_n}$ be a sample of 3D observed points and we suppose that each data point $x_i$ originates from one of  the $K$ components of the mixture model. 
Let $\textbf{x} = \pa{x_1, \ldots, x_n}$ be a sample of points, where each $x_i \in \mathbb{R}^3$ belongs to one component of the mixture. 
Given the sample $\textbf{x}$, we can estimate the parameter vector $\Theta$ of the mixture by using the Maximum-Likelihood estimator $\widehat{\Theta}$ given by 
\begin{eqnarray}\label{max_pb_em}
\widehat{\Theta} = \arg\max_{\Theta} L\pa{\textbf{x} ; \Theta}
\end{eqnarray}
where $L$ denotes the log-likelihood function defined as follows 
\begin{eqnarray}\label{mlequation}
L\pa{\textbf{x};\Theta}
= \sum_{i=1}^n \log h(x_i | \Theta)
= \sum_{i=1}^n \log\pa{\sum_{k=1}^K \pi_k \, f\pa{x_i\mid \theta_k}}
\end{eqnarray}
% 

%To estimate the parameter vector $\Theta$ of the mixture density function $h$, given the sample of data points $\textbf{x}$, we use
%the Maximum-Likelihood estimator $\widehat{\Theta}$ defined by 
%
%\begin{eqnarray}
%\widehat{\Theta} = \arg\max_{\Theta} L\pa{\textbf{x} ; \Theta}
%\end{eqnarray}
%
%where $L$ denotes the log-likelihood function defined as follows 
%
%\begin{eqnarray}\label{mlequation}
%L\pa{\textbf{x};\Theta}
%= \sum_{i=1}^n \log h(x_i | \Theta)
%= \sum_{i=1}^n \log\pa{\sum_{k=1}^K \pi_k \, f\pa{x_i\mid \theta_k}}
%\end{eqnarray}
% 
%The log-likelihood (\ref{mlequation}) is hard to maximize because of its non-linearity in many unknown parameters. 
The maximization problem \eqref{max_pb_em} is difficult to solve because of the non-linearity of the log-likelihood $L$ in many parameters. 
The Expectation-Maximization (EM) algorithm has been introduced in \cite{Dempster77} to solve this problem in the case of incomplete data 
(some extensions can be found in \cite{McLachlan2008}).

% 
%The main idea of EM is to associate at each $x_i$ a hidden variable $z_i = \pa{z_{i1}, \ldots, z_{iK}}$ containing the information of belonging to the clusters $\pa{C_1, \ldots, C_K}$ and defined by $z_{ik} = \ind_{\acc{x_i \in C_k}}$. 
%Values of variables $z_i$ define clusterings of data. 
The main idea of EM  is to introduce a hidden variable $z_i = \pa{z_{i1}, \ldots, z_{iK}}$ given the membership of each $x_i$ to the clusters $\pa{C_1, \ldots, C_K}$ . 
The hidden variable is defined by $z_{ik} = \ind_{\acc{x_i \in C_k}}$ and each  possible value leads to a clustering of data. 
Data points $\textbf{x}$ are called \textit{observed data} 
while $\textbf{z} = \pa{z_1, \ldots, z_n}$ are called \textit{missing data} since these values are not observed. 
We can then introduce $\textbf{z}$ in \eqref{mlequation} to obtain the so-called complete log-likelihood 
\begin{equation} \label{completed_loglik}
L_c(\textbf{x}, \textbf{z}; \Theta) = \sum_{i=1}^n \sum_{k=1}^K z_{ik} \pab{\log\pa{\pi_{k}} + 
\log\pa{f(x_i\mid\theta_k}}
\end{equation}

%We can then complete the log-likelihood (\ref{mlequation})
%by introducing $\textbf{z}$ to obtain the so-called complete log-likelihood 
%
%\begin{equation} \label{completed_loglik}
%L_c(\textbf{x}, \textbf{z}; \Theta) = \sum_{i=1}^n \sum_{k=1}^K z_{ik} \pab{\log\pa{\pi_{k}} + 
%\log\pa{f(x_i\mid\theta_k}}
%\end{equation}
%
The EM algorithm consists in maximizing the conditional expectation of the complete log-likelihood \eqref{completed_loglik} given the observation $\textbf{x}$ and the current parameters estimate $\Theta^{(j)}$ at iteration $j$. 
The quantity $t_{i\ell} = \dE\cro{z_{i\ell} \mid \textbf{x}, \Theta^{(j)}}$ is called membership probability of point $x_i$ to cluster $C_\ell$ and can be estimated with Bayes rule. 
Then, the EM algorithm consists of iterating te following two steps until convergence

\begin{itemize}
\item \textbf{E-step :} Compute the membership probabilities $t_{i\ell}$ given by
\begin{eqnarray}
t_{i\ell} = \dfrac{\pi_\ell^{(j)}\, f\pa{x_i\mid\theta_\ell^{(j)}}}{\sum_{k=1}^K \pi_k^{(j)} \, f\pa{x_i\mid\theta_k^{(j)}}}
\end{eqnarray}
 using the current parameters estimate $\Theta^{(j)} = \pa{\pi_1^{(j)}, \ldots, \pi_K^{(j)}, \theta_1^{(j)}, \ldots, ,\theta_K^{(j)}}$. 
\item \textbf{M-step :} Find parameters $\Theta^{(j+1)}$ that maximize the conditional expectation of $L_c$ given the observation $\textbf{x}$ and the current parameters estimate $\Theta^{(j)}$ 
\begin{eqnarray}
Q\pa{\Theta} = \sum_{i=1}^n \sum_{k=1}^K t_{ik} \pab{\log\pa{\pi_{k}} + 
\log\pa{f(x_i\mid\theta_k}}
\end{eqnarray}

\end{itemize}

\noindent
To start the iterative algorithm, an initial value $\widehat{\Theta}^{(0)}$ must be  determined, for example using the $K$-means clustering algorithm 
(see \cite{Bishop07} for example). 
%The EM algorithm is initialized using the $K$-means clustering algorithm (see \citet{Bishop07} for example) to determine the value of $\widehat{\Theta}^{(0)}$. 

% 
\noindent 
The maximization of $Q$ in the M-step is performed analytically by computing partial derivatives with respect to each parameter and involves the assumption on parameter $\sigma$ (for more details, see appendix \ref{appendixC_ell}). 
For each component $\ell \in [1, K]$ of the mixture, we obtain the following equations defining the estimators 

\begin{equation}
\left\lbrace 
\begin{array}{l}
\displaystyle \widehat{\mu}_\ell = \dfrac{1}{\sum_{i=1}^n t_{i\ell}} \, \cro{\sum_{i=1}^n t_{i\ell} \, X_i - \sum_{i=1}^n \dfrac{t_{i\ell} \, (X_i - \widehat{\mu}_\ell)}{\xi_{i, \ell}}} \\
\displaystyle \widehat{\Sigma}_\ell = \dfrac{3}{\sum_{i=1}^n t_{i\ell}} \, \sum_{i=1}^n t_{i\ell} (X_i - \widehat{\mu}_\ell) (X_i - \widehat{\mu}_\ell)^T \\
\displaystyle \widehat{\sigma}_\ell^2 = \dfrac{1}{\sum_{i=1}^n t_{i\ell}} \, \sum_{i=1}^n t_{i\ell} \pa{\xi_{i, \ell} - \overline{\xi}_{\ell}}^2 \\ \\
\displaystyle \widehat{\pi}_\ell = \dfrac{1}{n} \, \sum_{i=1}^n t_{i\ell} 
\end{array}
\right.
\end{equation}
with 
\begin{equation*}
\xi_{i, \ell} = \dMah{X_i}{\widehat{\mu}_\ell}{\widehat{\Sigma}_\ell} \quad \mbox{ and } \quad 
\overline{\xi}_{\ell} = \dfrac{1}{n} \sum_{i=1}^n \xi_{i, \ell}.
\end{equation*}
In addition, since estimators $\mu$ and $\Sigma$ depends to each other,  estimates of $\mu$ and $\Sigma$ are computed using the iterative backfitting procedure defined by \eqref{Def_muchap1}-\eqref{sigma_1}. 
The EM algorithm and the iterative backfiting procedure have a similar stopping criterion. 
Algorithms are stopped when the difference betwen consecutive estimated parameters values is small or when a maximal number of iterations has been reached. 
The behaviour of the estimation algorithm will be studied on simulated data in the following section. 

%%%%%%%%%%%%%%%%%%%%%%%%%%%%%%%%%%%%%%%%%%%%%%%%%%
\section{Experiments on simulated data}\label{section_experiments_simul}
%%%%%%%%%%%%%%%%%%%%%%%%%%%%%%%%%%%%%%%%%%%%%%%%%%

The objective of the experiments carried out in this section is to show that the proposed estimation methods (backfitting BF and EM) behave as expected in simulations. 
Parameters of the simulated distributions have been chosen such that the assumption on parameter $\sigma$ is satisfied. 
Samples of a random variable $X$ of density $f$ have been simulated in dimensions $d=2$ and $d=3$ using  the decomposition $X = \mu + \Sigma^{1/2} \, W \, U$, a rejection sampling method on the random variable $W$ and an inverse transform method on the random vector $U$. 
The case $d=2$ has been considered to compare the estimation method to a standard algorithm. %pour pouvoir se comparer a un algo standard performant deja programmé. 
More precisely, the ellipse estimation method BF is compared to \cite{Fitz95} implemented in the open source OpenCV library \cite{OpenCV}. 
Numerical results will be analyzed in the case of a single ellipsoid in section \ref{ssection_simul1}  and for a mixture model in section \ref{ssection_simulK}. 
 
%%%%%%%%%%%%%%%%%%%%%%%%%%%%%%%%%%%%%%%%%%%%%%%%%%
\subsection{The case of the single ellipsoid}\label{ssection_simul1}
%%%%%%%%%%%%%%%%%%%%%%%%%%%%%%%%%%%%%%%%%%%%%%%%%%

% 
The first set of experiments deals with the modeling of a single ellipsoid. 
We first consider random vectors $X_d$, with $d=3$ with ellipsoidal density $f$ introduced in \eqref{def_densite}.  
For each experiment, we simulate $N=200$ independant samples of realizations of the random vector $X_d$.  
The parameters are set to $\mu_{(2)} = (0, 0)^T$, $\small \Sigma_{(2)} =
\left( \begin{array}{ccc}
100^2 & 0 \\ 
0 & 50^2
\end{array} \right)$, $\sigma_{(2)} = 0.01$ for $d=2$ and $\mu_{(3)} = (0, 0, 0)^T$, $\small \Sigma_{(3)} =
\left( \begin{array}{ccc}
100^2 & 0 & 0 \\ 
0 & 50^2 & 0 \\
0 & 0 & 50^2
\end{array} \right)$, $\sigma_{(3)} = 0.01$ for $d=3$.

% resultats ellipsoide complete  
Figure \ref{ellipsoide_1} shows boxplots describing the $N$ estimations of parameters $\mu_{(3)}$, $\Sigma_{(3)}$ and $\sigma_{(3)}$ obtained with the BF algorithm for different sample sizes. 
For the shape matrix, we compute the error $E(\Sigma) = \pa{\sum_{i,j=1}^{d } \pa{\wh{\Sigma}_{i, j} - \Sigma_{i, j}}^2}^{1/2}$. 
We note that the estimated values are close to the expected ones for a sufficiently high number of observations. 
It is clear that when the sample size increases, the accuracy of the results is improved and the variability decreases. 
Let us note that the same remarks can be made for the results obtained in the case $d=2$. 
%\vspace{1ex}

\begin{figure}[h!]
\centering 
\begin{tabular}{ccc}
\includegraphics[scale=0.3]{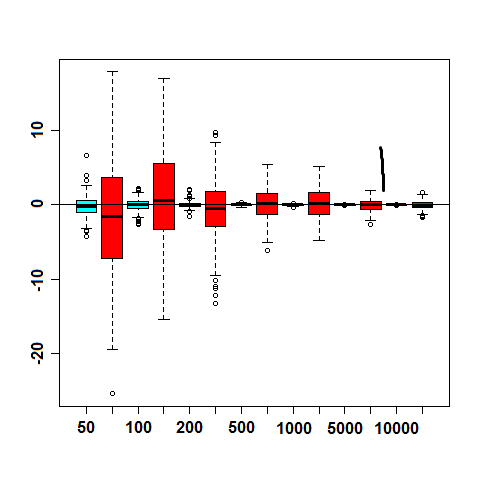} &
\includegraphics[scale=0.3]{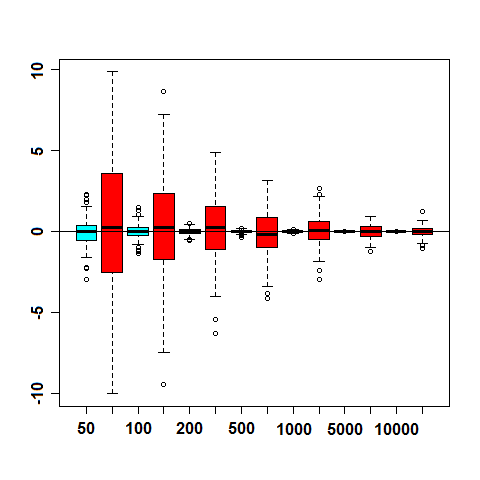} &
\includegraphics[scale=0.3]{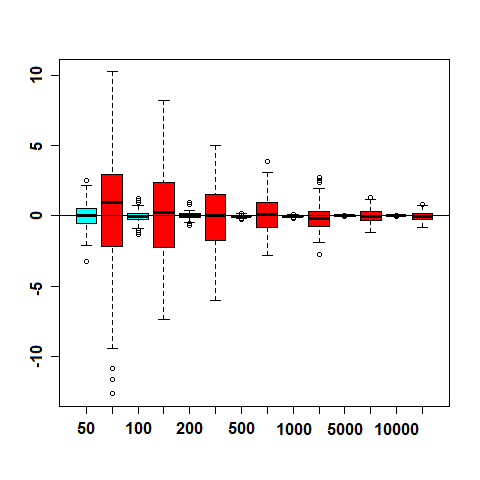} \\
\includegraphics[scale=0.23]{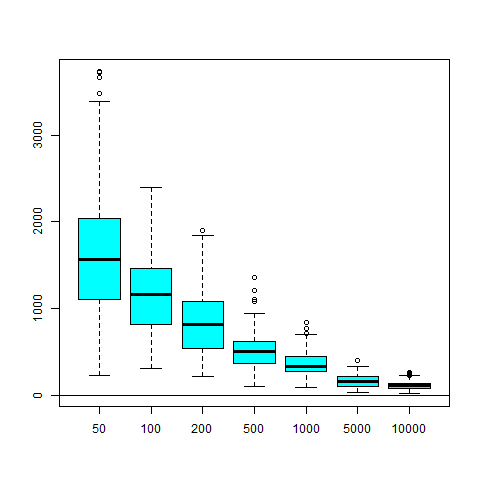} &
\includegraphics[scale=0.3]{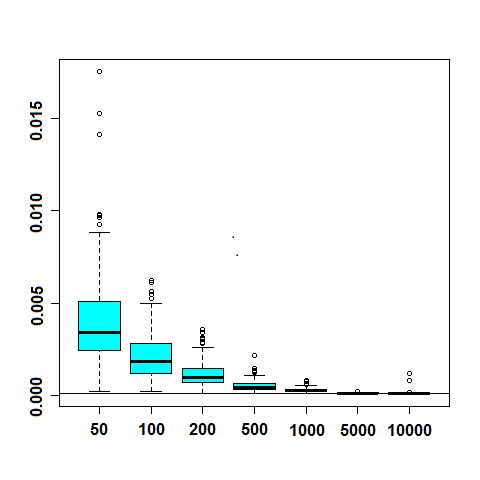} &
\includegraphics[scale=0.25]{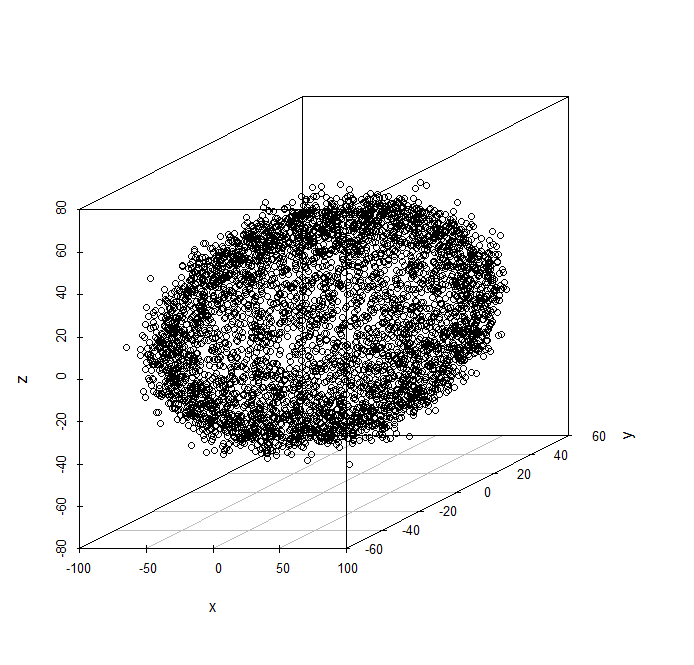}
\end{tabular}
\caption{On the first row, boxplots of estimates of $\mu_x$, $\mu_y$ and $\mu_z$. On the second row, boxplots of $E(\Sigma)$, estimates of  $\sigma$ and a sample of size $n=1000$.}
\label{ellipsoide_1}
\end{figure}

%\begin{figure}[h!]
%\centering 
%\begin{tabular}{ccc}
%\includegraphics[scale=0.25]{images/experim/simul/ellipsoide_entiere/bp_cov12.png} &
%\includegraphics[scale=0.25]{images/experim/simul/ellipsoide_entiere//bp_cov13.png} &
%\includegraphics[scale=0.25]{images/experim/simul/ellipsoide_entiere/bp_cov23.png} \\
%\includegraphics[scale=0.25]{images/experim/simul/ellipsoide_entiere/bp_s2.png} &
%\includegraphics[scale=0.2]{images/experim/simul/ellipsoide_entiere/ellipsoide_entiere.png} &
% 
%\end{tabular}
%\caption{On the first row, estimates of $\Sigma_{12}$, $\Sigma_{13}$ and $\Sigma_{23}$. On the second row, estimates of $\sigma$ and a sample of size $n=1000$.}
%\label{ellipsoide_2}
%\end{figure}

% comparaison estimateurs centre 
In section \ref{ssection_iterestim}, we mentionned that the variance of the estimator $\overline{X}_n^\ast$ is smaller than the variance of the empirical mean $\overline{X}_n$. 
A comparison between these two estimators is presented on Figure \ref{comp_center} for $d=3$. 
We observe that the variability of the center estimates is widely smaller with BF than with the empirical mean. 
% \vspace{1ex}

\begin{figure}[h!]
\centering 
\begin{tabular}{ccc}
\includegraphics[scale=0.3]{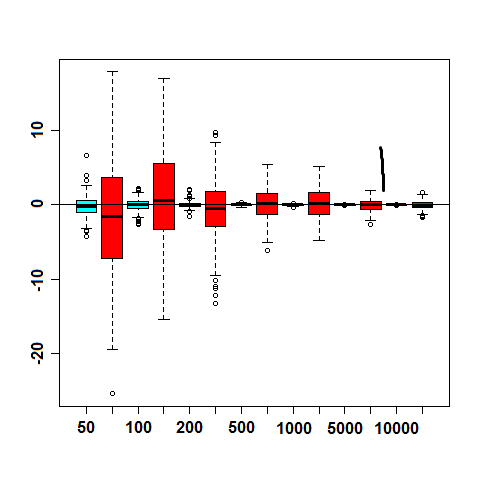} &
\includegraphics[scale=0.3]{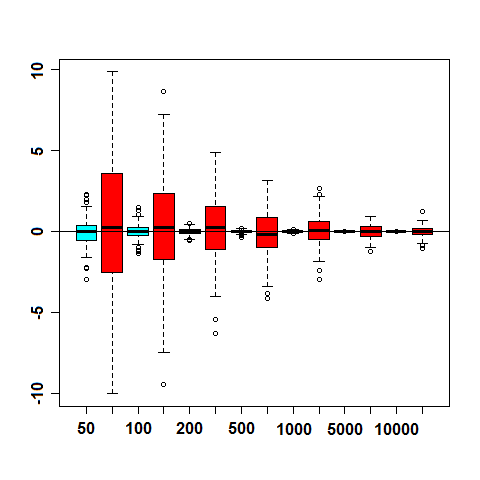} &
\includegraphics[scale=0.3]{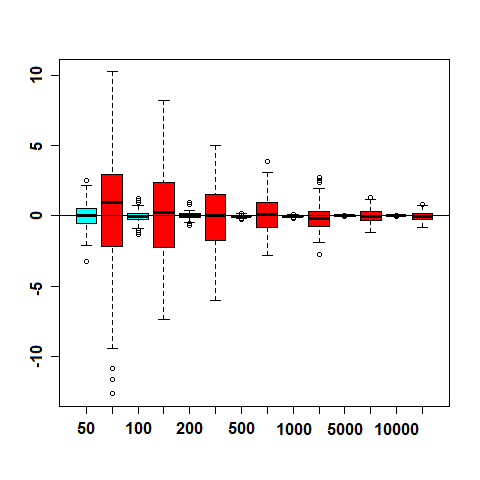}
\end{tabular}
\caption{Boxplot of estimates of parameters $\mu_x$, $\mu_y$ and $\mu_z$ for different sample sizes $n$. The BF method is plotted in blue and the empirical mean in red.}
\label{comp_center}
\end{figure}

Finally, we compare the proposed estimation method to the ellipse fitting algorithm \cite{Fitz95} implemented in the open source OpenCV library \cite{OpenCV}. 
This reference method is widely used in computer vision to fit an ellipse to a set of image edges. 
We simulate samples from the density \eqref{def_densite} with $d=2$ and with parameters $\mu = \pa{0, 0}$, $\sigma = 0.01$ and $\Sigma = \left( \begin{array}{cc}
100^2 & 50^2 \\ 
50^2 & 50^2
\end{array} \right)$. 
Obtained results are plotted on Figure \ref{simul_compocv}.
For the centre estimation, we observe that the proposed method offers estimations as accurate as the reference one. 
Moreover, the variability of the estimates is of the same order of magnitude for the two compared algorithms.  
For the shape matrix estimation, we note that the BF algorithm offers slightly better estimates than the reference one. 
However, the variability of the estimates computed with the BF method is larger.

\begin{figure}[h!]
\centering 
\begin{tabular}{ccc}
\includegraphics[scale=0.24]{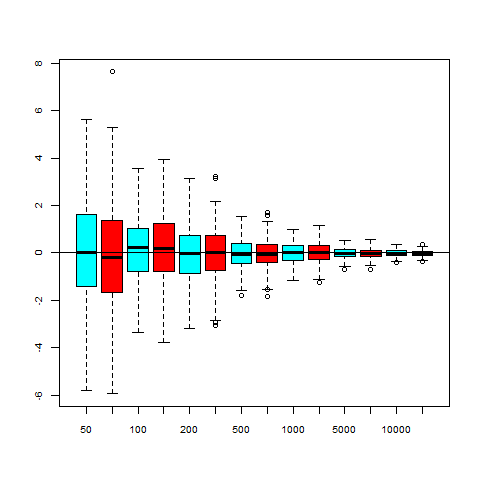} &
\includegraphics[scale=0.24]{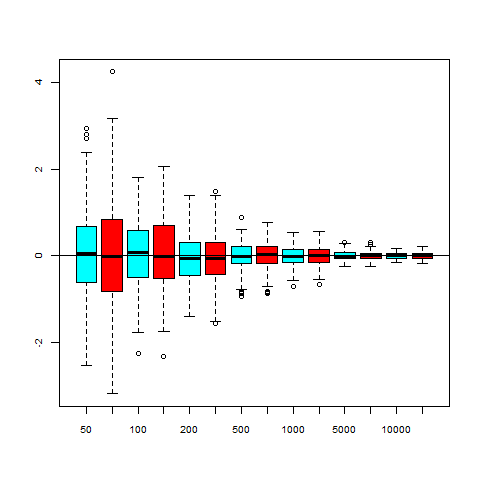} &
\includegraphics[scale=0.24]{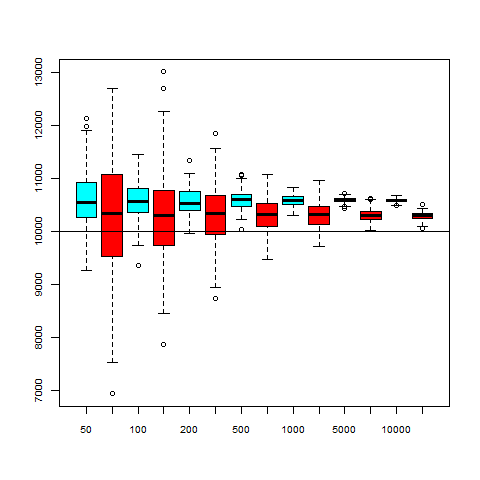} \\
\includegraphics[scale=0.24]{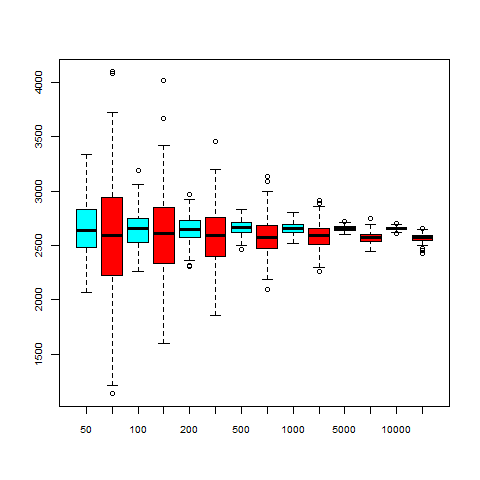} &
\includegraphics[scale=0.24]{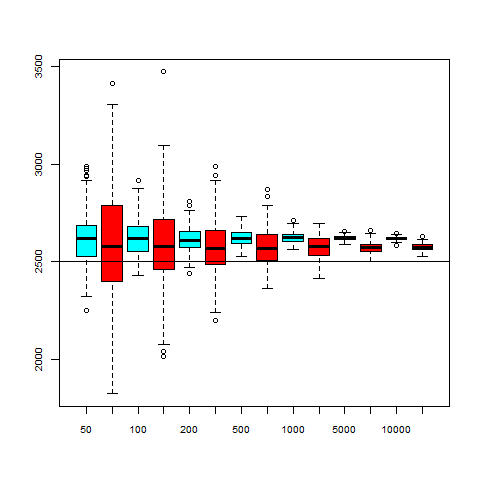} &

\end{tabular}
\caption{On the first row, boxplots of estimates of $\mu_x$, $\mu_y$ and $\Sigma_{11}$. On the second row, estimates of $\Sigma_{12}$ and $\Sigma_{22}$. The BF method is plotted in red and the reference method is plotted in blue.}
\label{simul_compocv}
\end{figure}

\subsection{A mixture model with three components}\label{ssection_simulK}
%%%%%%%%%%%%%%%%%%%%%%%%%%%%%%%%%%%%%%%%%%%%%%%%%%

The aim of this second set of experiments is to show that the EM algorithm is able to separate and model several overlapping ellipsoids. 
The EM algorithm is initialized with a parameter vector $\Theta^{(0)}$ computed from a classification obtained with the K-means clustering algorithm. 
We consider an ellipsoidal mixture model made up of 3 overlapping components $E_1$, $E_2$ and $E_3$. 
We simulate a sample of size $n=3000$. 
Parameters of the model are given in Table \ref{table_params}. 
The three shape matrices are the same for all the components and set to 
$\left( \begin{array}{ccc}
1000^2 & 0 & 0 \\ 
0 & 500^2 & 0 \\
0 & 0 & 500^2
\end{array} \right)$. 
% \vspace{1ex}

\begin{table}[h!]
\centering
\caption{Exact parameters values.}
\label{table_params}
\begin{tabular}{|c|c|c|c|c|c|}
\hline 
  & $\mu_x$ & $\mu_y$ & $\mu_z$ & $\sigma$ & $\pi$ \\ 
\hline 
E1 & 0 & 0 & 0 & 0.01 & 1/3 \\ 
\hline 
E2 & 1800 & 0 & 0 & 0.01 & 1/3 \\ 
\hline 
E3 & 3600 & 0 & 0 & 0.01 & 1/3 \\ 
\hline  
\end{tabular}
\end{table}

Results are presented on Figures \ref{simul_melange3_poids}, \ref{simul_melange3_datares}, \ref{simul_melange_3_1} and \ref{simul_melange_3_2}.  
We observe that the variability of the estimations decreases when the sample size increases. 
We remark an important number of bad estimations when the sample size is small. 
%Diagonal terms of the shape matrix  are less accurately estimated than in the case of a single ellipsoid. 
The estimation of the shape matrix is more difficult, especially for $E_2$,  because of the intersection of the distributions. 
Observations located between two ellipsoids contributes to the estimation of the two component.   

\begin{figure}[h!]
\centering 
\begin{tabular}{ccc}
\includegraphics[scale=0.3]{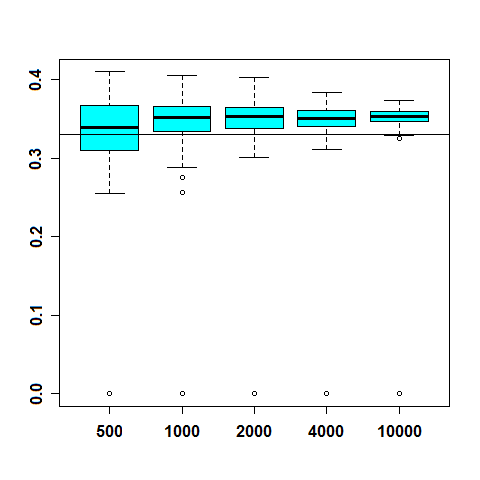} &
\includegraphics[scale=0.3]{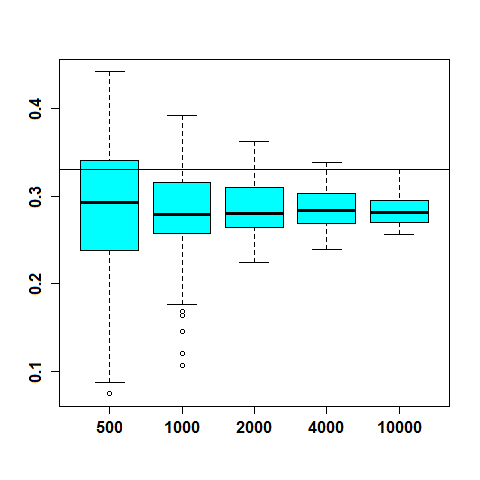} &
\includegraphics[scale=0.3]{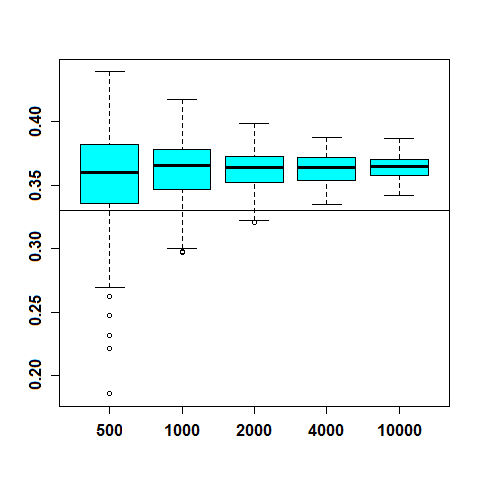}
\end{tabular}
\caption{Boxplots of estimates of the mixture weights $\pi_1$, $\pi_2$ and $\pi_3$.}
\label{simul_melange3_poids}
\end{figure}

\begin{figure}[h!]
\centering 
\begin{tabular}{cc}
\includegraphics[scale=0.3]{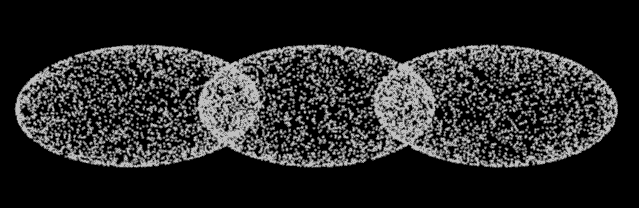} &
\includegraphics[scale=0.3]{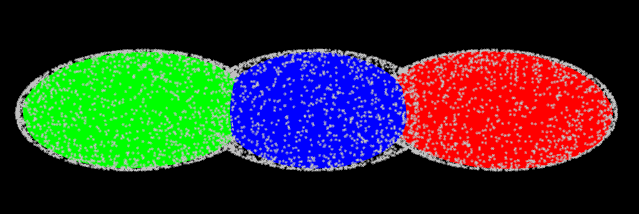}
\end{tabular}
\caption{A sample of size $n=10000$ and the estimated model.}
\label{simul_melange3_datares}
\end{figure}

\begin{figure}[h!]
\centering 
\begin{tabular}{ccc}
\includegraphics[scale=0.3]{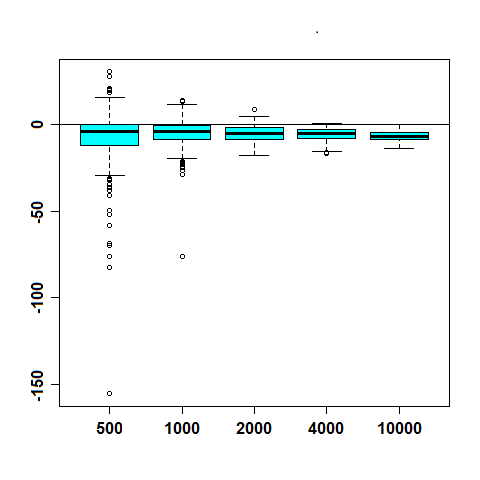} &
\includegraphics[scale=0.3]{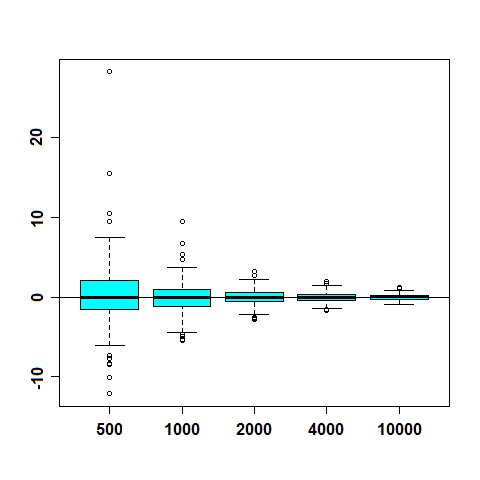} &
\includegraphics[scale=0.3]{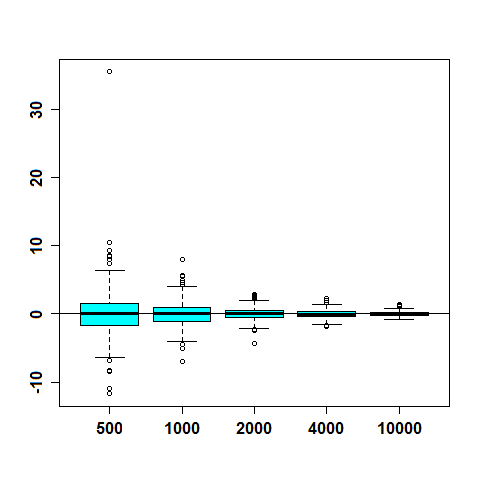} \\
\includegraphics[scale=0.23]{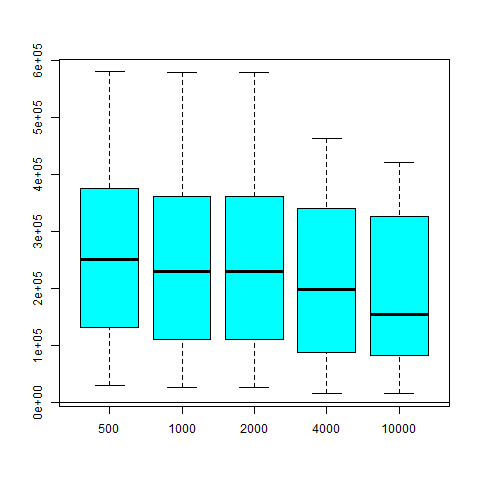} &
 &
 
\end{tabular}
\caption{Estimates of $E_1$ parameters. On the first row, boxplot of estimates of $\mu_x$, $\mu_y$ and $\mu_z$. 
On the second row, estimates of $E(\Sigma)$. 
%On the second row, estimates of $\Sigma_{11}$, $\Sigma_{22}$ and %$\Sigma_{33}$. On the third row, estimates of $\Sigma_{12}$, $\Sigma_{13}$ and $\Sigma_{23}$.
}
\label{simul_melange_3_1}
\end{figure}

\begin{figure}[h!]
\centering 
\begin{tabular}{ccc}
\includegraphics[scale=0.3]{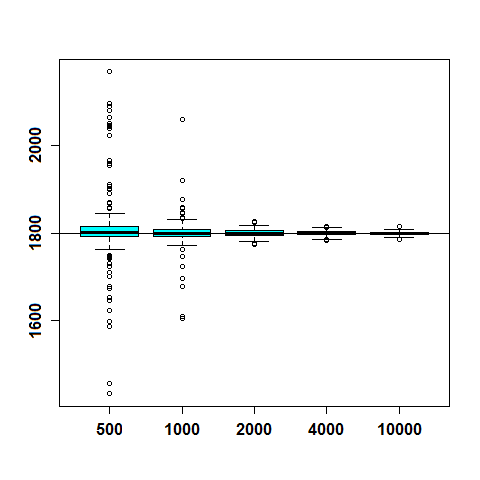} &
\includegraphics[scale=0.3]{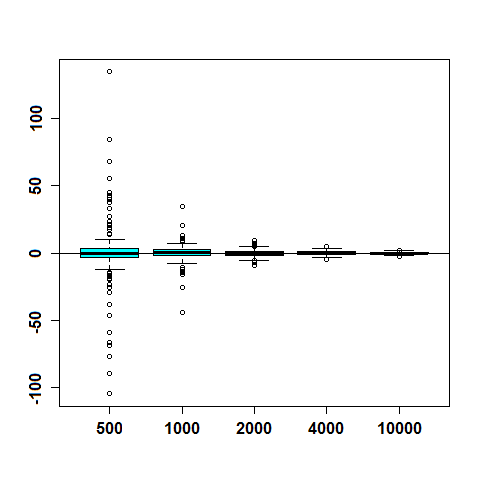} &
\includegraphics[scale=0.3]{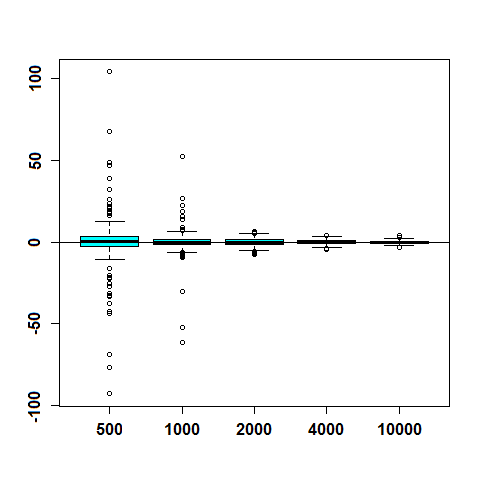} \\
\includegraphics[scale=0.23]{bp_mat_l2_1.png} &
 &
  
\end{tabular}
\caption{Estimates of $E_2$ parameters. 
On the first row, boxplot of estimates of $\mu_x$, $\mu_y$ and $\mu_z$. 
On the second row, estimates of $E(\Sigma)$. 
%On the second row, estimates of $\Sigma_{11}$, $\Sigma_{22}$ and $\Sigma_{33}$. 
%On the third row, estimates of $\Sigma_{12}$, $\Sigma_{13}$ and $\Sigma_{23}$.
}
\label{simul_melange_3_2}
\end{figure}

\newpage 

%%%%%%%%%%%%%%%%%%%%%%%%%%%%%%%%%%%%%%%%%%%%%%%%%%
\section{Applications}\label{section_applications}
%%%%%%%%%%%%%%%%%%%%%%%%%%%%%%%%%%%%%%%%%%%%%%%%%%

In this section, we experiment our ellipsoidal mixture model on real data.
The objective is to model a person in 3D. 
Data are acquired with two calibrated depth sensor located on the opposite sides of a room. 
The two point clouds are fused thanks to a 3D registration technnique  \cite{Besl92}. 
The data are then segmented to keep only points belonging to the person.
We apply the proposed algorithm with $K=2$ and $K=3$ components in the mixture model. 
Results are given in Figure \ref{exp_person}. 
The model with $K=2$ components gives a better representation of the point cloud. 
The head is modelled by a nearly spherical shape while the body is modelled by a single ellipsoid. 
In the case of three components, the third ellipsoid is not meaningfull. 

\begin{figure}[h!]
\centering 
\begin{tabular}{ccc}
\includegraphics[scale=0.2]{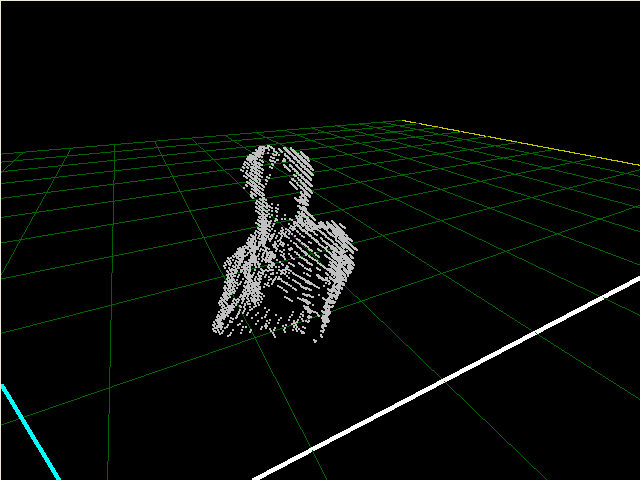} &
\includegraphics[scale=0.2]{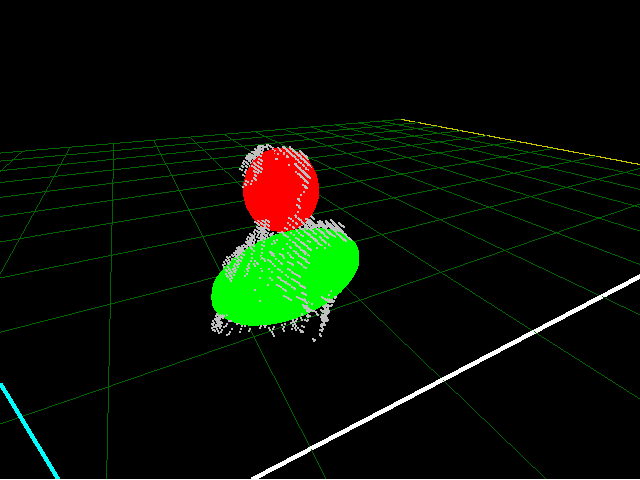} &
\includegraphics[scale=0.2]{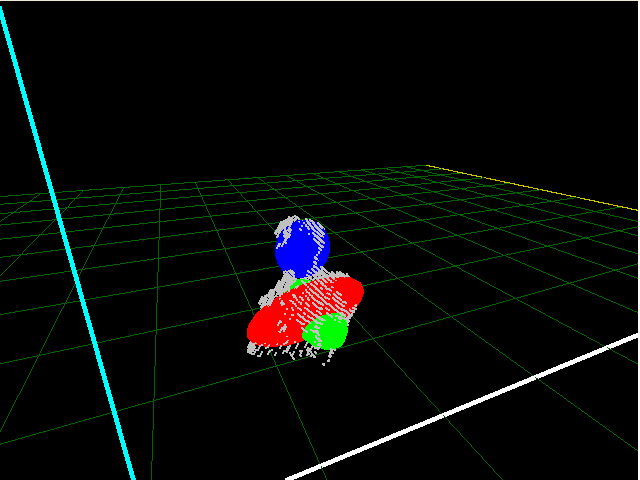} \\
\includegraphics[scale=0.2]{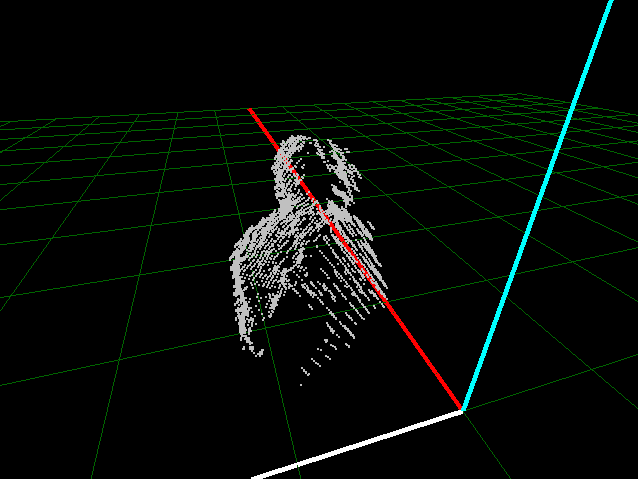} &
\includegraphics[scale=0.2]{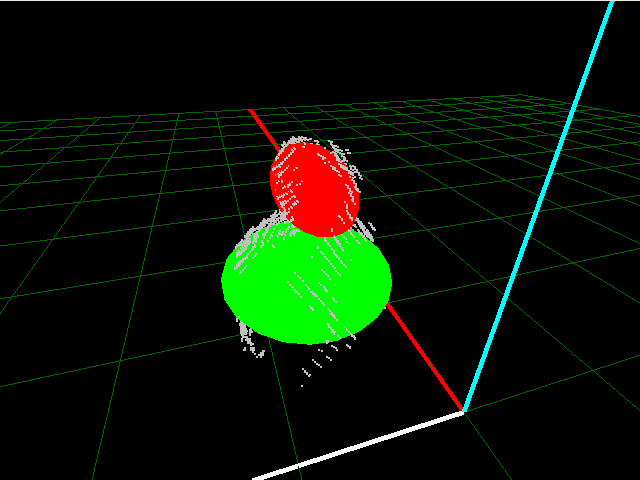} &
\includegraphics[scale=0.2]{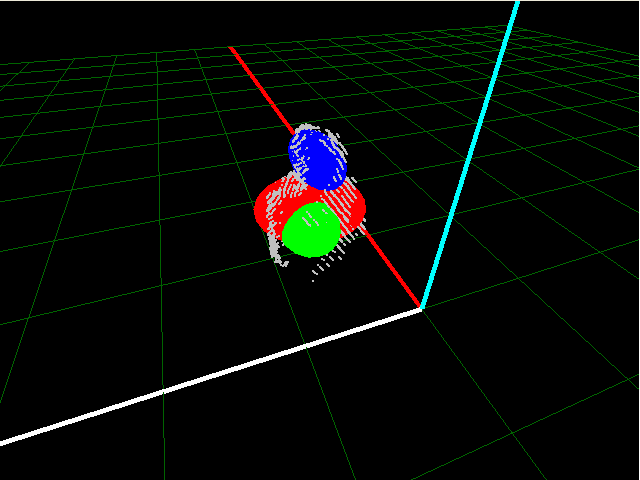} \\
(a) & (b) & (c) 
\end{tabular}
\caption{On the first column, 3D point clouds representing a standing person. On the second column, the estimated model with $K=2$ components, and on the third column, the estimated model with $K=3$ components. }
\label{exp_person}
\end{figure}

%%%%%%%%%%%%%%%%%%%%%%%%%%%%%%%%%%%%%%%%%%%%%%%%%%
\section*{Conclusion}\label{section_conclusion}
%%%%%%%%%%%%%%%%%%%%%%%%%%%%%%%%%%%%%%%%%%%%%%%%%%

% travail effectué  
In this work, we proposed a new ellipsoidal mixture model for 3D data modelling.  
The mixture model is based on a new pdf belonging to the family of elliptical distributions and that models points spread around an ellipsoidal surface. 
The parameters of the mixture model based on this density are estimated with an EM based algorithm. 
The convergence of the method has been established theoretically. 
Our algorithm is simple to implement and produces accurate estimations. 
We obtained results comparable to a state of the art ellipse fitting method. 

% test statistique pou discriminer sphere / ellipsoide 
Experiments revealed that the analysis of the shape matrix $\Sigma$ could be a promising way to determine if a given set of points is a sphere or an ellipsoid. 
An adequate statistical test could therefore be set up to classify objects. 

\begin{appendix}

\section{Proofs and technical calculus}\label{sec::proofs}
\subsection{Proofs of Lemma \ref{thm_ellipsoid}  and \ref{tthm_Xetoile}}\label{sec::proofs::lemma}
\begin{proof}[Proof of Lemma \ref{thm_ellipsoid}]
The proof is straightforward using the decomposition $X = \mu + \Sigma^{1/2} \, W \, U$. 
The random vector $U$ is distributed on the $d$-dimensional unit sphere, which implies that  $\norm{U} = 1$, $\dE\cro{U} = 0$ and $\Var\cro{U} = d^{-1} \, I_d$. 
We therefore deduce that $\dE\cro{X} = \mu$ and 
$\Var\cro{X} = d^{-1} \, \dE\cro{W^2} \, \Sigma = d^{-1} \, \SigmaEt$. 
We can conclude the first part of the proof by remarking that for $q\geq 0$,\ 
$\dE\cro{W^q} = J_{d+q-1}J_{d-1}^{-1}$.

\noindent 
Then, to prove the second part, we set 
$$\xi = \sqrt{(X - \mu)^T\SigmaEt^{-1}(X - \mu)}, $$
and we can show that $\xi = (\dE\cro{W^2})^{-1/2} W$, leading to
$$\Var\pa{\dMah{X}{\mu}{\SigmaEt}}\ = \ 
\dE\pa{\xi^2} - \pa{\dE\cro{\xi}}^2 
 =\ 1 -  \dfrac{(\dE\cro{W})^2}{\dE\cro{W^2}},$$
which closes the proof of the lemma. 
\end{proof}

\begin{proof}[Proof of Lemma \ref{tthm_Xetoile}]
The proof is based on the decomposition $X = \mu + \Sigma^{1/2} \, W \, U$. 
The random vector $X^\ast$ can therefore be rewritten as $X^\ast = \mu + \Sigma^{1/2} \, U \, \pa{W - 1}$. 
Then, since random variables $W$ an $U$ are independant with $\dE\cro{U} = 0$, we deduce that that $\dE\cro{X^\ast} = \mu$.

\noindent 
This result enables us to compute the variance: $\Var\cro{X^\ast} = \dE\cro{\pa{X^\ast - \mu} \pa{X^\ast - \mu}^T}$. 
Given that variables $W$ and $U$ are independent with $\Var\cro{U} = d^{-1} \, I_d$ and by  remarking that $\dE\cro{W^q} = J_{d+q-1} \, J_{d-1}^{-1}$, we obtain $\Var\cro{X^\ast}= d^{-1} \, \dE\cro{(W - \sqrt{\dE\cro{W^2}})^2} \, \Sigma$. 
\end{proof}
\subsection{Proof of Theorem \ref{thm_conv_direct}}\label{appendixB_ell}
%%%%%%%%%%%%%%%%%%%%%%%%%%%%%%%%%%%%%%%%%%%%%%%
\begin{proof}[Proof of Theorem \ref{thm_conv_direct}]
Since $\dE[\norm{X_1}^4] < \infty$, using the law of the iterated logarithm of Hartman and Winters (see \cite{Stout}), we straight forwardly obtain  the first two results. 
We now focus on the convergence properties of the estimator $\widehat\sigma^2_n$ presented in Section \ref{ssection_directestim}.
For this purpose, we first establish the convergence of $\Sigmachap^{-1}$ towards  $\SigmaS^{-1}$. 
The proof is standard and we resume the main steps of the proof proposed by N'Guyen and  Saracco in \cite{Nguyen10}.

From the Riccati equation for matrix inversion given for example in \cite{DUF1997} (page 96), we can write the following decomposition 
\[
\Sigmachap^{-1}\ =\ \SigmaS^{-1} -
\SigmaS^{-1}\pa{\Sigmachap - \SigmaS}\SigmaS^{-1}
\ +\ R_n, 
\]
where $R_n$ is given by 
\[
R_n \ =\ \SigmaS^{-1}\pa{\SigmaS - \Sigmachap}\Sigmachap^{-1}
\pa{\SigmaS - \Sigmachap}\SigmaS^{-1}.
\]
Then, we immediately deduce that 
\[
\norm{\Sigmachap^{-1} - \SigmaS^{-1}}  \leq  
\lambda^2_{\max}(\SigmaS^{-1}) \norm{\Sigmachap - \SigmaS}\ +\ \norm{R_n}
\]
and 
\[
\norm{R_n}   \leq   \lambda_{\max} \, (\Sigmachap^{-1})\,
\lambda^2_{\max}(\SigmaS^{-1}) \norm{\SigmaS - \Sigmachap}^2.
\]
Let us recall that for any square matrix $A$ positive definite, we have $\lambda_{\max}(A^{-1}) = \left( \lambda_{\min}(A)\right)^{-1}$.
Consequently, since $\lambda_{\min}(\Sigma) > 0$ then $\lambda_{\min}(\SigmaS) > 0$. 
In addition, since $\Sigmachap$ converges almost surely towards $\SigmaS$, then  $\lambda_{\min}(\Sigmachap) > 0$.
We therefore have  
\begin{eqnarray}
\norm{R_n} & = & O\pa{\norm{\SigmaS - \Sigmachap}^2}\quad\mbox{a.s.}
\end{eqnarray}
Thanks to the upper bound $\norm{\Sigmachap - \SigmaS} = O(a_n)$, with $a_n = \sqrt{(\log\log n)/n}$, established in Section \ref{ssection_directestim}, we can write 
\begin{eqnarray}\label{MajInvSigmachap}
\norm{\Sigmachap^{-1} - \SigmaS^{-1}} & = & O\pa{a_n}\quad\mbox{a.s.}
\end{eqnarray}

Let us now study the convergence of $\widehat\sigma^2_n$ that can be rewritten under the form 
$A_n - B_n^2$ with 
$$A_n = \dfrac{1}{n}\sum_{j=1}^{n}\widehat\xi_j^2
\quad \quad\mbox{ and }\quad \quad
B_n = \dfrac{1}{n}\sum_{j=1}^{n}\widehat\xi_j, $$
where we recall that 
$\widehat\xi_j = \sqrt{(X_j - \Xbarn)^T\Sigmachap^{-1} (X_j - \Xbarn)}$.
In order to study the convergence of $A_n$, it is rewritten under the form 
\begin{eqnarray} \label{eqn_a_n}
A_n  & = & A_{1,n}\ + \ A_{2,n}
\ -\  (\Xbarn - \mu)^T \SigmaS^{-1} (\Xbarn - \mu) 
\end{eqnarray}
with 
\begin{align*}
A_{1,n} & =  \dfrac{1}{n}\sum_{j=1}^{n} (X_j - \Xbarn)^T
\pa{\Sigmachap^{-1} - \SigmaS^{-1}}(X_j - \Xbarn)\noindent\\
A_{2,n} & =  \dfrac{1}{n}\sum_{j=1}^{n} (X_j - \mu)^T\SigmaS^{-1} (X_j - \mu).
\end{align*}
We can easily show that 
$\sum_{j=1}^{n} \norm{X_j - \Xbarn}^2 = O(n)$ a.s.,
and by using \eqref{MajInvSigmachap}, we obtain 
$\abs{A_{1,n}} = O(a_n)$ a.s. 
On the other hand, we have 
\[
A_{2,n} = \dfrac{1}{n}\sum_{j=1}^{n} \dfrac{J_2(\sigma)}{J_4(\sigma)} W_j^2  \tendps  1
\]
and since the $(W_j)$ have a finite moment of order 4, the law of the iterated logarithm of Hartman and Winters \cite{Stout} can be applied, leading to $\abs{A_{2,n} - 1} = O(a_n)$ a.s.

Finally, as $\norm{\Xbarn - \mu} = O(a_n)$ a.s., we get from \eqref{eqn_a_n} the almost sure convergence of $A_n$ towards 1 and the upper bound  

\begin{equation}\label{MajAn}
\abs{A_n - 1} = O(a_n)\quad\mbox{a.s.}
\end{equation}
Let us now study the convergence of $B_n$. 
To this aim, $B_n$ is rewritten under the form 
\begin{equation}\label{DecompositionBn}
B_n  = B_{1,n}\ + \ B_{2,n}\ + \ B_{3,n}
\end{equation}
with 
\begin{align*}
B_{1,n}  = & \dfrac{1}{n}\sum_{j=1}^{n} \xi_j\\
B_{2,n}  = & \dfrac{1}{n}\sum_{j=1}^{n} 
\pa{\sqrt{(X_j  - \mu)^T\Sigmachap^{-1} (X_j - \mu)} - \xi_j} \\
B_{3,n}  = & \dfrac{1}{n}\sum_{j=1}^{n} \pa{\widehat\xi_j - \sqrt{(X_j - \mu)^T\Sigmachap^{-1} (X_j - \mu)}} 
\end{align*}
where we recall that 
$\xi_j = \sqrt{(X_j - \mu)^T\SigmaS^{-1} (X_j - \mu)}$.
We immediately deduce that for any $1\leq j\leq n$,
\[
(\lambda_{\min}(\SigmaS^{-1}))^{1/2} \norm{X_j - \mu}
\leq \xi_j \leq
(\lambda_{\max}(\SigmaS^{-1}))^{1/2} \norm{X_j - \mu}.
\]
In addition, by remarking that $\xi_j$ can be written 
$\xi_j = \pa{\dE \left[W^2\right]}^{-1/2}W_j$, we deduce that  
\begin{eqnarray}\label{MajB1n}
B_{1,n} \ \tendps\ 
\dfrac{\dE[W]}{(\dE \left[ W^2 \right])^{1/2}}
& \mbox{and} & 
\abs{B_{1,n} - \dfrac{\dE[W]}{\left(\dE \left[W^2\right]\right)^{1/2}}}
= O\pa{a_n}\ \mbox{a.s.}
\end{eqnarray}
By using the equality $\sqrt{a} - \sqrt{b} = (a-b) (\sqrt{a} + \sqrt{b})^{-1}$,
we easily obtain 
\begin{align*}
\abs{B_{2,n}} & \leq  \dfrac{1}{n}\sum_{j=1}^{n} 
\dfrac{\abs{(X_j  - \mu)^T\Sigmachap^{-1} (X_j - \mu)  - \xi_j^2}}
{\sqrt{(X_j  - \mu)^T\Sigmachap^{-1} (X_j - \mu)} + \xi_j} \\
& \leq  \dfrac{1}{n}\sum_{j=1}^{n} 
\dfrac{\abs{(X_j  - \mu)^T\pa{\Sigmachap^{-1}-\SigmaS^{-1}}(X_j - \mu)}}{\xi_j}.
 \nonumber
\end{align*}
Since $\xi_j \geq (\lambda_{\min}(\SigmaS^{-1}))^{1/2} \norm{X_j - \mu}$, we deduce that 
\[
\abs{B_{2,n}}  \leq   
(\lambda_{\max}(\SigmaS))^{1/2} \norm{\Sigmachap^{-1}-\SigmaS^{-1}}
\dfrac{1}{n}\sum_{j=1}^{n} \norm{X_j - \mu}.
\]
Finally, as  
$\sum_{j=1}^{n} \norm{X_j - \mu} = O(n)$ a.s., we have  
\begin{equation}\label{MajB2n}
\abs{B_{2,n}}= O\pa{a_n}\quad\mbox{a.s.}
\end{equation}
Thanks to the triangular inequality (derived from the Mahanobis distance), we get 
\[
\abs{B_{3,n}} \leq \dfrac{1}{n}\sum_{j=1}^{n} 
\sqrt{(\Xbarn - \mu)^T\Sigmachap^{-1} (\Xbarn - \mu)} 
 \leq  \sqrt{\lambda_{\max}(\Sigmachap^{-1})} \norm{\Xbarn - \mu}.
\]
Finally, since on the one hand $\norm{\Xbarn - \mu} = O\pa{a_n}$ a.s.
and on the other hand, $\Sigmachap^{-1}$ convergerges almost surely towards $\SigmaS^{-1}$
with $\lambda_{\min}(\SigmaS) > 0$, we deduce that  
\begin{eqnarray}\label{MajB3n}
\abs{B_{3,n}} & = & O\pa{a_n}\quad\mbox{a.s.}
\end{eqnarray}
By combining \eqref{MajB1n}, \eqref{MajB2n} and \eqref{MajB3n} 
with \eqref{DecompositionBn}, we get 
\begin{eqnarray}\label{MajBn}
B_n \ \tendps\ 
\dfrac{\dE[W]}{(\dE [W^2])^{1/2}}
& \mbox{and} & 
\abs{B_n - \dfrac{\dE[W]}{(\dE [W^2])^{1/2}}}
= O\pa{a_n}\ \mbox{a.s.}
\end{eqnarray}
The result on the convergence of $\widehat\sigma_n^2 = A_n - B_n^2$ is obtained by combining the results \eqref{MajAn} and \eqref{MajBn}.
Indeed, we first obtain  
\[
\widehat\sigma_n^2 \tendps
1 - \dfrac{(\dE(W))^2}{\dE(W^2)} = 1 - \dfrac{J_3^2}{J_2J_4}
= \widetilde\sigma^2 .
\]
Then, by considering the decomposition 
\[
\widehat\sigma_n^2 - \widetilde\sigma^2 = \pa{A_n - 1} - \pa{B_n - \gamma}\pa{B_n + \gamma}, 
\]
where we set $\gamma = (\dE [W^2])^{-1/2} \dE[W]$,
and by using the upper bound $\abs{B_n + \gamma} = O(1)$ a.s., 
we deduce that 
\[
\abs{\widehat\sigma_n^2 - \widetilde\sigma^2} = O\pa{a_n}\ \mbox{a.s.}
\]
which closes the proof of the convergence of $\widehat\sigma_n^2$. 
\end{proof}
%%%%%%%%%%%%%%%%%%%%%%%%%%%%%%%%%%%%%%%%%%%%%%%

%%%%%%%%%%%%%%%%%%%%%%%%%%%%%%%%%%%%%%%%%%%%%%%
\section{Some important results on the normalization constant of the ellipsoidal density}
%\lipsum[4]
\label{appendixA_ell}
%%%%%%%%%%%%%%%%%%%%%%%%%%%%%%%%%%%%%%%%%%%%%%%

This appendix deals with the computation of the normalization constant $C_d$ of the new ellipsoidal density $f$ introduced in \eqref{def_densite}. 
We first establish the following technical lemma which is usefull to simplify computations. 

%%%%%%%%%%%%%%%%%%%%%
%  Lemme Technique
\def\drho{{\mbox{\rm d}\rho}}
\def\dt{{\mbox{\rm d}t}}
\begin{lem}\label{LemCalcul_J}
%%%%%%%%%%
Let us set for $q \in \dN$ and $\alpha> 0$, 
\begin{eqnarray}\label{Def_Jd}
J_q(\alpha)\ =\ \int_{0}^{\infty} t^{q} \exp\pa{-(t-1)^2 / \pa{2\alpha^2}} \dt.
\end{eqnarray}
Then $J_0(\alpha) = \alpha\sqrt{2\pi}\pa{1 - \Phi(-1 / \alpha)}$,
$J_1(\alpha) = J_0(\alpha) + \alpha^2 \exp\pa{-1 / \pa{2\alpha^2}}$,
and for $q \geq 2$,
\begin{eqnarray}\label{Val_Jd}
J_q(\alpha)\ =\  J_{q-1}(\alpha) +  (q-1) \alpha^2  J_{q-2}(\alpha)
\end{eqnarray}

\noindent 
where $\Phi$ denotes the gaussian cumulative distribution function. 
\end{lem}

\begin{proof}
Let us set $v(t) = -(t-1)^2 / \pa{2\alpha^2}$.
Then, $v^\prime(t) = - (t-1) / \alpha^2$ and we can rewrite 
\eqref{Def_Jd} under the form
\begin{eqnarray}
\notag J_q(\alpha) & = & \int_{0}^{\infty} \pa{t^{q} - t^{q-1}} 
\exp\pa{v(t)}\dt + J_{d-1}(\alpha)\\
& = & -\alpha^2 \int_{0}^{\infty} t^{q-1} v^\prime(t) \exp\pa{v(t)}\dt 
+ J_{q-1}(\alpha)\label{Calcul_Jd}
\end{eqnarray}

Finally, \eqref{Val_Jd} immediately follows using an integration by parts.
To close the proof, it is clear that $J_0(\alpha) = \alpha\sqrt{2\pi}\pa{1 - \Phi(-1 / \alpha)}$ and we get
 $J_1(\alpha)$ using \eqref{Calcul_Jd} with $d=1$. 
\end{proof}

Let us now consider the constant $C_d$ which satisfies 
\[
I = C_d \int_{\mathbb{R}^d} \exp\pa{- \dfrac{1}{2 \sigma^2} \pa{\sqrt{\pa{x-\mu}^T \Sigma^{-1} \pa{x-\mu}} - 1}^2} dx = 1.
\]
The substitution $y := \Sigma^{-1/2}\pa{x - \mu}$ leads to 
\[
C_d \int_{\mathbb{R}^d} \det{\Sigma^{1/2}} \exp\pa{- \dfrac{1}{2 \sigma^2} \pa{\norm{y} - 1}^2} dy = 1.
\]
Then, a polar coordinates substitution leads to   
\begin{align*}
I & =  C_d \, \det{\Sigma}^{1/2} \, \dfrac{2 \pi^{d/2}}{\Gamma\pa{d/2}} \int_0^{+\infty} \rho^{d-1} \exp\pa{- \dfrac{\pa{\rho - 1}^2}{2 \sigma^2}} d\rho \\
& =  C_d \, \det{\Sigma}^{1/2} \, \dfrac{2 \pi^{d/2}}{\Gamma\pa{d/2}} J_{d-1}\pa{\sigma} = 1.
\end{align*}
The expression of the normalization constant is 
\[
C_d = \dfrac{\Gamma\pa{d/2}}{2 \pi^{d/2} \, \mid \Sigma \mid^{1/2} J_{d-1}\pa{\sigma}}.
\]
For the first values of $d$, the expressions of the constant are given by 
\begin{align*}
C_1 & =  \dfrac{1}{2 \sigma \sqrt{2 \pi} \det{\Sigma}^{1/2} \cro{1 - \Phi\pa{-1 / \sigma}}} \\
C_2 & =  \dfrac{1}{(2 \pi)^{3/2} \det{\Sigma}^{1/2} \sigma\cro{1 - \Phi\pa{-1 / \sigma}}} \\
C_3 & =  \dfrac{1}{2 (2 \pi)^{3/2} \det{\Sigma}^{1/2} \sigma (1 + \sigma^2) \cro{1 - \Phi\pa{-1/\sigma} + \dfrac{\sigma \exp\pa{- 1 / (2 \sigma^2)}}{(1 + \sigma^2) \sqrt{2 \pi}}}} 
\end{align*}
where $\Phi$ denotes the gaussian cumulative distribution function.

We observe that the parameter $\sigma$ plays a particular role in these expressions. 
Indeed, under the assumption that the parameter $\sigma$ is small enough, terms $\Phi(-1/\sigma)$ and $\frac{t}{\sqrt{2 \pi} (1 + \sigma^2)} e^{(-1 / (2 \sigma^2))}$ can be ignored (cf Appendix \ref{appendixA_bis_ell}). 
In that case,  we can approximate the normalization constant and thus obtain 
\begin{align*}
C_1 & =  \dfrac{1}{2 \sigma \sqrt{2 \pi} \det{\Sigma}^{1/2}} \\
C_2 & =  \dfrac{1}{(2 \pi)^{3/2} \det{\Sigma}^{1/2} \sigma} \\
C_3 & =  \dfrac{1}{2 (2 \pi)^{3/2} \, \det{\Sigma }^{1/2} \sigma (1 + \sigma^2)}.
\end{align*}

%%%%%%%%%%%%%%%%%%%%%%%%%%%%%%%%%%%%%%%%%%%%%%%
\section{Influence of parameter $\sigma^2$}
%\lipsum[4]
\label{appendixA_bis_ell}
%%%%%%%%%%%%%%%%%%%%%%%%%%%%%%%%%%%%%%%%%%%%%%%

This appendix concerns the influence of parameter $\sigma^2$ on several quantities.

Figure \ref{courbes_approx} represents the values of terms  $\Phi(-1/\sigma)$ and $\frac{t}{\sqrt{2 \pi} (1 + \sigma^2)} e^{(-1 / (2 \sigma^2))}$ with respect to $\sigma$. 
We observe that these quantities are close to zero for small values of $\sigma$.

\begin{figure}[h!]
\centering 
\includegraphics[scale=0.5]{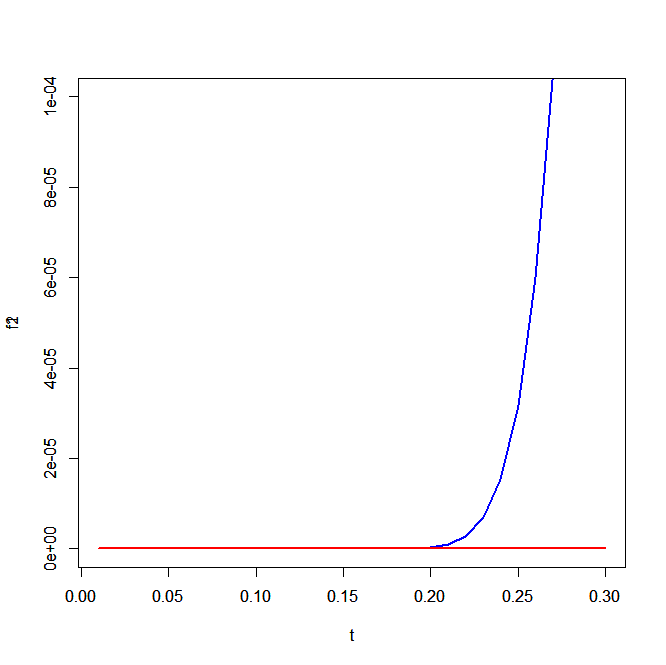}
\caption{Values of terms $\Phi(-1/\sigma)$ (blue) and $\frac{t}{\sqrt{2 \pi} (1 + \sigma^2)} e^{(-1 / (2 \sigma^2))}$ (red) with respect to $\sigma$.}
\label{courbes_approx} 
\end{figure}

Figure \ref{courbes_approx_s} shows the values of terms $\dfrac{1 + 6 \sigma^2 + 3 \sigma^4}{1 + \sigma^2}$, $\dfrac{1 + 3 \sigma^4}{1 + 7 \sigma^2 + 9 \sigma^4 + 3 \sigma^6}$ and $\dfrac{1 + 3 \sigma^2}{1 + \sigma^2}$ with respect to $\sigma$. 
We can therefore quantify the difference between the given terms and the value $1$. 

\begin{figure}[h!]
\centering 
\includegraphics[scale=0.25]{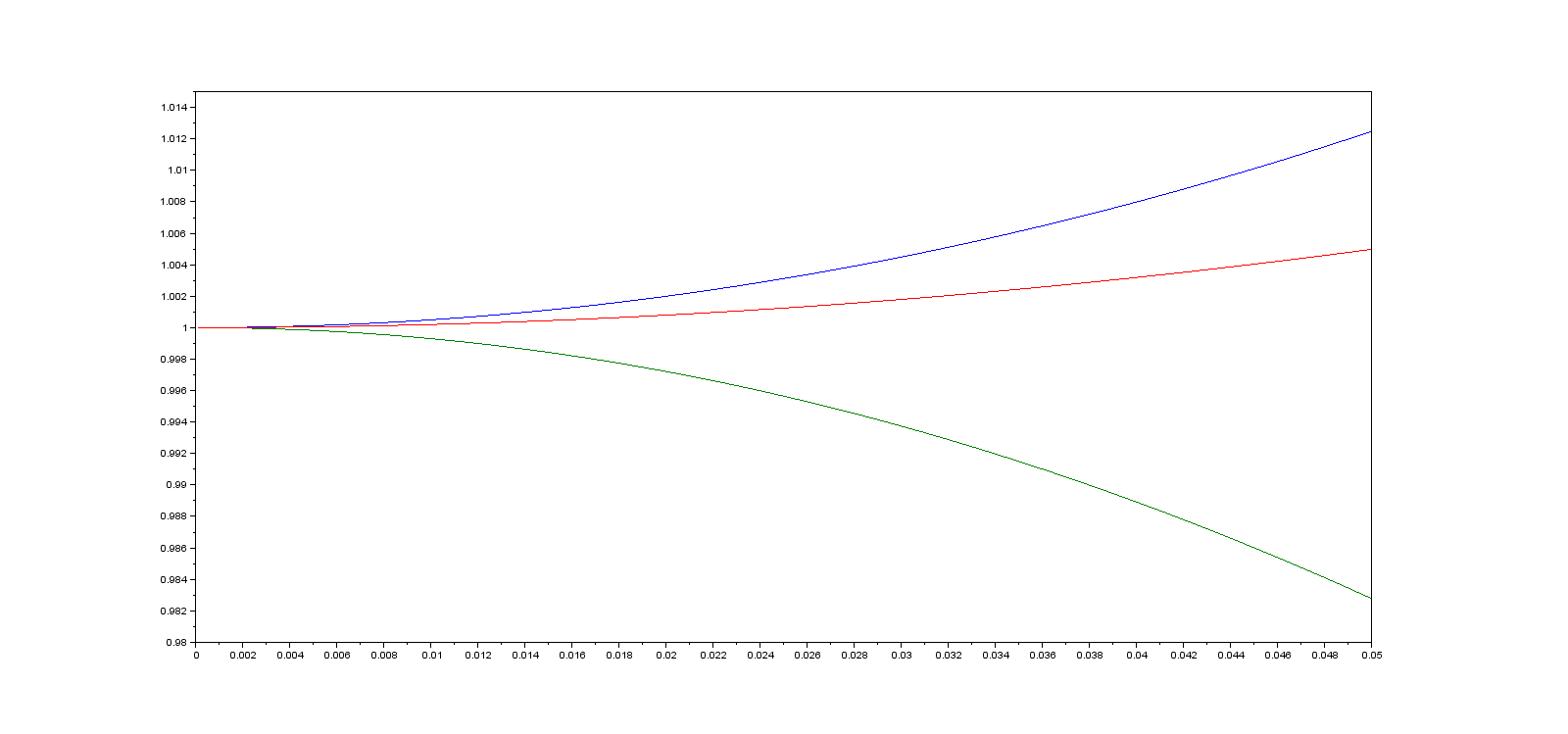}
\caption{Values of terms $\dfrac{1 + 6 \sigma^2 + 3 \sigma^4}{1 + \sigma^2}$ (blue), $\dfrac{1 + 3 \sigma^4}{1 + 7 \sigma^2 + 9 \sigma^4 + 3 \sigma^6}$ (green) and  $\dfrac{1 + 3 \sigma^2}{1 + \sigma^2}$ (red) with respect to $\sigma$.}
\label{courbes_approx_s} 
\end{figure}

%%%%%%%%%%%%%%%%%%%%%%%%%%%%%%%%%%%%%%%%%%%%%%%
\section{Justification of the M-step in the EM algorithm}
%\lipsum[4]
\label{appendixC_ell}
%%%%%%%%%%%%%%%%%%%%%%%%%%%%%%%%%%%%%%%%%%%%%%%

In this appendix, we justify the choice of the estimators used in the M-step of the EM algorithm described in Section \ref{ssection_estimmm}. 
To maximize the conditional expectation of the completed log-likelihood of parameters $\Theta = \pa{\pi_1, \ldots, \pi_K, \theta_1, \ldots, \theta_K}$ with $\theta_k = \pa{\mu_k, \Sigma_k, \sigma_k}$, we compute partial derivatives with respect to each parameter of the quantity   
\[
Q\pa{\Theta} = \sum_{i=1}^n \sum_{k=1}^K t_{ik} \pa{\log\pa{\pi_k} + \log\pa{f(x_i, \theta_k}}, 
\]
where 
$$f(x \mid \theta_k) = \dfrac{(2 \pi)^{-3/2} \, \det{\Sigma_k}^{-1/2}}{2 \, \sigma_k (1 + \sigma_k^2)} \, \exp\pa{- \dfrac{1}{2 \sigma_k^2} \pa{\dMah{x}{\mu_k}{\Sigma_k} - 1}^2}$$
and with $\sigma$ sufficiently small. The partial derivative of $Q$ with respect to $\mu_k$ is   
\[
\dfrac{\partial Q}{\partial \mu_k} = \dfrac{\Sigma_k^{-1}}{\sigma_k^2} \, \sum_{i=1}^n t_{ik} \, (x_i - \mu_k) \cro{1 - \dfrac{1}{\dMah{x_i}{\mu_k}{\Sigma_k}}}.
\]
The partial derivative of $Q$ with respect to $\Sigma_k^{-1}$ is   
\[
\dfrac{\partial Q}{\partial \Sigma_k^{-1}} = \sigma^2 \, \Sigma_k - \dfrac{1}{n} \, \sum_{i=1}^n t_{ik} \, (x_i - \mu_k)(x_i - \mu_k)^T \cro{1 - \dfrac{1}{\dMah{x_i}{\mu_k}{\Sigma_k}}}.
\]
The partial derivative of $Q$ with respect to $\sigma_k^2$ is   
\begin{equation}\label{deriv_sigma}
\dfrac{\partial Q}{\partial \sigma_k} = - n \, \dfrac{(1 + 3 \, \sigma_k^2)}{\sigma_k (1 + \sigma_k^2)} + \dfrac{1}{\sigma_k^3} \, \sum_{i=1}^n  t_{ik} \, \pa{\dMah{x}{\mu_k}{\Sigma_k} - 1}^2.
\end{equation} 
The maximization problem cannot be solved analytically. 
Under the assumption that $\sigma_k$ is small enough, we can write $(1 + 3 \, \sigma_k^2) \,  / \, (1 + \sigma_k^2) \simeq 1$ and the partial derivative \eqref{deriv_sigma} can be simplified. 
However, for shake of simplicity, we prefer estimate the shape matrix $\Sigma_k$ by the usual estimator $n^{-1} \, \sum_{i=1}^n (x_i - \mu_k)(x_i - \mu_k)^T$. 
This choice is also motivated by our experiments.

We can easily show that the mixture weights $\left\lbrace w_{ik} \right\rbrace_{i=1..n}^{k=1..K}$ can be estimated using the classical estimators employed for gaussian mixture models.  
Finally, these computations lead to the estimators presented in Section \ref{ssection_estimmm}. 

\end{appendix}

%\bibliography{smj-template}

\begin{thebibliography}{}

\end{thebibliography}


\begin{thebibliography}{00}

\bibitem[Ahn et~al. (2002)]{Ahn02} 
Sung Joon Ahn, Wolfgang Rauh, Hyung Suck Cho, and H-J Warnecke Orthogonal
distance fitting of implicit curves and surfaces. 
\textit{IEEE Transactions on Pattern Analysis
and Machine Intelligence}, 24(5):620–638, 2002. 


\bibitem[Besl (1992)]{Besl92} 
Paul J. Besl. A Method for Registration of 3-D Shape.
\textit{IEEE Trans. on Pattern Analysis and Machine Intelligence}, 
Los Alamitos, CA, USA, vol. 14, no 2, 1992, p. 239–256.

\bibitem[Bishop (2007]{Bishop07} 
Christopher M. Bishop. \textit{Pattern Recognition and Machine Learning (Information Science and Statistics)}, 
Springer. 

\bibitem[Brazey et~al.(2014)]{Brazey14}
D. Brazey and B. Portier (2014). A new spherical mixture model for head detection in depth images. 
\textit{SIAM Journal on Imaging Science}, 
Vol. 7, No. 4, pp. 2423-2447.


\bibitem[Burt et~al. (2019)]{Burt19} 
Burt, A, Disney, M, Calders, K. Extracting individual trees from lidar point clouds using treeseg.
\textit{Methods Ecol Evol}, 10: 438– 445, 2019.



\bibitem[Dempster et~al.(1977)]{Dempster77} 
A. P. Dempster, N. M. Laird and D. B. Rubin. Maximum Likelihood from Incomplete Data via the EM Algorithm. 
\textit{Journal of the Royal Statistical Society}, 
Vol. 39, pp. 1-38, 1977.

\bibitem[Du et~al. (2018)]{Du18} 
X. Du, M. H. Ang, S. Karaman and D. Rus. A General Pipeline for 3D Detection of Vehicles.
\textit{ IEEE International Conference on Robotics and Automation (ICRA)},  pp. 3194-3200, 2018. 



\bibitem[Duflo et~al.(1997)]{DUF1997} 
M. Duflo. \textit{Random Iterative Models}, 
Springer Verlag, Berlin.

\bibitem[Fischler et~al.(1981)]{Fischler81} 
M. A. Fischler and R. C. Bolles. Random sample consensus : a paradigm for model fitting with applications to image analysis. 
\textit{Communications of the ACM}, vol. 24, pp. 381-395. 

\bibitem[Fitzgibbon et~al.(1995)]{Fitz95} 
Andrew W. Fitzgibbon and R.B.Fisher. A Buyer’s Guide to Conic Fitting. 
\textit{Proc.5th British Machine Vision Conference}, Birmingham, pp. 513-522. 

\bibitem[Greggio et~al.(2012)]{Greggio12} 
N. Greggio, A. Bernardino, C. Laschi, P. Dario and J. Santos-Victor. Fast Estimation of Gaussian Mixture Models for Image Segmentation. 
\textit{Machine Vision and Applications}, Vol. 23, pp. 773-789.

\bibitem[Jian et~al.(2005)]{Jian05} 
B. Jian and B. C. Vemuri. A Robust Algorithm for Point Set Registration Using Mixture of Gaussians. 
\textit{International Conference on Computer Vision}, Vol. 2, pp. 1246-1251. 

\bibitem[Kesäniemi et~al. (2017)]{Kesaniemi17} 
Martti Kesäniemi and Kai Virtanen Direct least square fitting of hyperellipsoids. 
\textit{IEEE Transactions on Pattern Analysis and Machine Intelligence}, 40(1):63–76, 2017. 

\bibitem[Li et~al. (2004)]{Li04} 
Qingde Li and John G Griffiths Least squares ellipsoid specific fitting. 
\textit{Geometric
Modeling and Processing},  pages 335–340, 2004. 




\bibitem[Li et~al. (2019)]{Li19} 
Li, L., Sung, M., Dubrovina, A., Yi, L. and Guibas, L. J. Supervised fitting of geometric primitives to 3d point clouds. 
\textit{Proceedings of the IEEE/CVF Conference on Computer Vision and Pattern Recognition}, pp. 2652-2660, 2019. 

\bibitem[Li et~al. (2020)]{Li20} 
Li Y, Li W, Darwish W, Tang S, Hu Y, Chen W. Improving Plane Fitting Accuracy with Rigorous Error Models of Structured Light-Based RGB-D Sensors.
\textit{Remote Sensing}, 12(2):320, 2020.


\bibitem[Liu et~al.(2014)]{Liu14} 
J. Liu and Z. Wu. An adaptive approach for primitive shape extraction from point clouds. 
\textit{Optik - International Journal for Light and Electron Optics}, 125(9), pages 2000-2008.

\bibitem[Martins et~al.(2008)]{Martins08} 
D.A. Martins, A.J. Neves and A.J. Pinho. Real-time generic ball recognition in RoboCup domain.
\textit{Proc. of the 3rd International Workshop on Intelligent Robotics}, pages 37-48.

\bibitem[McLachlan et~al.(2000)]{McLachlan2000}  
G. J. McLachlan and D. Peel. Finite Mixture Models. 
\textit{Wiley Series in Probability and Statistics}.

\bibitem[McLachlan et~al.(2008)]{McLachlan2008} 
G. J. McLachlan and K. Thriyambakam. The EM Algorithm and Extensions. 
\textit{WILEY-Interscience}.  

\bibitem[Muirhead (1982)]{Muirhead1982} 
R. J. Muirhead. Aspects of multivariate statistical theory, Wiley. 

\bibitem[Nahangi et~al. (2019)]{Nahangi19} 
M. Nahangi, T. Czerniawski, Carl T. Haas, S. Walbridge. Pipe radius estimation using Kinect range cameras.
\textit{Automation in Construction}, vol. 99, pp. 197-205, 2019.


\bibitem[Nguyen et~al.(2010)]{Nguyen10} 
T.M.N. Nguyen and J. Saracco. Estimation r\'ecursive en r\'egression inverse par tranches (sliced inverse regression).  
\textit{Journal de la Soci\'et\'e Fran\c caise de Statistique}, vol. 151 pp. 19-46.

\bibitem[N\'{u}\~{n}ez et~al.(2009)]{Nunez09} 
P. N\'{u}\~{n}ez, P. Drews Jr, R. Rocha, M. Campos and J. Dias. Novelty detection and 3D shape retrieval based on gaussian mixture models for autonomous surveillance robotics,. 
\textit{International Conference on Intelligent Robots and Systems}, pp. 4724-4730. 

\bibitem[OpenCV ()]{OpenCV} 
OpenCV library. \textit{http://opencv.org/}. 

\bibitem[Parr et~al. (2022)]{Parr22} 
Parr B., Legg M. Alam, F. Analysis of Depth Cameras for Proximal Sensing of Grapes.
\textit{Sensors}, no. 11: 4179, 2022. 

\bibitem[Qi et~al. (2017)]{Qi17} 
Qi, C. R., Yi, L., Su, H., and Guibas, L. J. Pointnet++: Deep hierarchical feature learning on point sets in a metric space.
\textit{Advances in neural information processing systems}, vol. 30, 2017.


\bibitem[Rabbani (2006)]{Rabbani06} 
T. Rabbani. Automatic reconstruction of industrial installations using point clouds and images. 
\textit{NCG Nederlandse Commissie voor Geodesie Netherlands Geodetic Commission}. 


\bibitem[Romanengo et~al. (2022)]{Romanengo22} 
C. Romanengo, A. Raffo, Y. Qie, N. Anwer, B. Falcidieno. Fit4CAD: A point cloud benchmark for fitting simple geometric primitives in CAD objects.
\textit{Computers and Graphics}, vol. 102, pp 133-143, 2022. 


\bibitem[Shafiq et~al.(2001)]{Shafiq01} 
M.S. Shafiq, S.T. Tümer, S. Turgut and H.C. Güler. Marker detection and trajectory generation algorithms for a multicamera based gait analysis system,. 
\textit{Mechatronics}, 11(4), pages 409-437.


\bibitem[Sharma et~al. (2020)]{Sharma20} 
Sharma, G., Liu, D., Maji, S., Kalogerakis, E., Chaudhuri, S., Měch, R. PARSENET: A Parametric Surface Fitting Network for 3D Point Clouds. 
\textit{European Conference on Computer Vision (ECCV), Lecture Notes in Computer Science}, pp 261–276, 2020. 

\bibitem[Shi et~al. (2021)]{Shi21} 
M. Han, J. Kan and Y. Wang. Ellipsoid Fitting Using Variable Sample Consensus and Two-Ellipsoid-Bounding-Counting for Locating Lingwu Long Jujubes in a Natural Environment. 
\textit{IEEE Access}, vol. 7, pp. 164374-164385, 2019. 

\bibitem[Shi et~al. (2021)]{Shi21b} 
Shi Y, Zhao G, Wang M, Xu Y, Zhu D. An Algorithm for Fitting Sphere Target of Terrestrial LiDAR.
\textit{Sensors},21(22):7546, 2021. 

\bibitem[Sommer et~al. (2020)]{Sommer20} 
C. Sommer, Y. Sun, E. Bylow, and D.
Cremers. PrimiTect: Fast Continuous Hough Voting for
Primitive Detection.
\textit{ICRA}, 2020. 


\bibitem[Stauffer et~al.(1999)]{Stauffer99} 
C. Stauffer, E. Grimson. Adaptive background mixture models for real-time tracking. 
\textit{Computer Vision and Pattern Recognition}.  

\bibitem[Stout (1974)]{Stout} 
W.F. Stout. Almost sure convergence. 
\textit{A Series of Monographs and Textbooks in Probability and
Mathematical Statistics}.

\bibitem[Sun et~al. (2019)]{Sun2019} 
B. Sun and P. Mordohai. Oriented point sampling
for plane detection in unorganized point clouds.
\textit{ICRA},21(22), 2019. 

\bibitem[Thom (1955)]{Thom} 
A. Thom,. A statistical examination of the megalithic sites in Britain. 
\textit{Journal of the Royal Statistical Society, Series A (General)}, pages 275-295.



\bibitem[Yang et~al. (2020)]{Yang20} 
Yang, B., Liu, C., Li, B., Jiao, J., Ye, Q. Prototype Mixture Models for Few-Shot Semantic Segmentation.
\textit{ECCV}, vol 12353, 2020. 

\bibitem[Zhang et~al.(2007)]{Zhang07} 
X. Zhang, J. Stockel, M. Wolf, P. Cathier, G. McLennan, E.A. Hoffman and M. Sonka.  A New Method for Spherical Object Detection and Its Application to Computer Aided Detection of Pulmonary Nodules in CT Images. 
\textit{Medical Image Computing and Computer-Assisted Intervention}, vol. 4791, pages 842-849.

% ajout references DB 


\bibitem[Zhao et~al. (2021)]{Zhao21} 
Zhao, Mingyang and Jia, Xiaohong and Ma, Lei and Qiu, Xinlin and Jiang, Xin and Yan, Dong‐Ming. Robust Ellipsoid-specific Fitting via Expectation Maximization.
\textit{Conference BMVC}, 2021. 

\bibitem[Zuo et~al. (2019)]{Zuo19} 
J. Zuo, Z. Jia, J. Yang and N. Kasabov Moving Target Detection Based on Improved Gaussian Mixture Background Subtraction in Video Images.
\textit{IEEE Access}, vol. 7, pp. 152612-152623, 2019.  














\end{thebibliography}

\bibliographystyle{apalike}

\end{document}